\ifdefined\JMLversion

\documentclass{ws-jml}
\usepackage{hyperref}
\usepackage[numbers,square,comma]{natbib}

\else 

\documentclass[12pt]{article}
\usepackage{fullpage}
\usepackage[backref]{hyperref}
\hypersetup{colorlinks=true, linkcolor=red,citecolor=blue,urlcolor=blue}

\fi   

\usepackage{amsmath,amsthm,amssymb,MnSymbol}

\usepackage{enumitem}
\usepackage{lineno}

\setlist[enumerate]{labelindent=0.6em,leftmargin=*}

\theoremstyle{theorem}
\newtheorem{theorem}{Theorem}
\newtheorem{lemma}[theorem]{Lemma}
\newtheorem{proposition}[theorem]{Proposition}
\newtheorem{corollary}[theorem]{Corollary}

\theoremstyle{definition}
\newtheorem{definition}[theorem]{Definition}

\newtheorem{remark}[theorem]{Remark}

\newcommand{\out}{\textit{out}}

\newcommand{\N}{\mathbb{N}}
\newcommand{\PTIME}{\mathsf{P}}

\newcommand{\eval}{\mathit{eval}}
\newcommand{\Log}{\mathit{Log}_{>1}}
\newcommand{\Logold}{\mathit{Log}}
\newcommand{\floor}[1]{\lfloor #1 \rfloor}

\newcommand{\bit}{\mathit{bit}}

\renewcommand{\le}{\leqslant}
\renewcommand{\leq}{\leqslant}
\renewcommand{\ge}{\geqslant}
\renewcommand{\geq}{\geqslant}

\newcommand{\NE}{\mathsf{NE}}
\newcommand{\NTIME}{\mathsf{NTIME}}
\newcommand{\NEXP}{\mathsf{NEXP}}
\newcommand{\EXP}{\mathsf{EXP}}

\renewcommand{\P}{\mathsf{P}}
\newcommand{\PH}{\mathsf{PH}}
\newcommand{\AC}{\mathsf{AC}}
\newcommand{\NP}{\mathsf{NP}}
\newcommand{\PSPACE}{\mathsf{PSPACE}}
\newcommand{\Ppoly}{\mathsf{P/poly}}
\newcommand{\PHpoly}{\mathsf{PH/poly}}
\newcommand{\SIZE}{\mathsf{SIZE}}

\newcommand{\PHP}{\mathit{PHP}}

\newcommand{\T}{\mathsf{T}}
\newcommand{\PV}{\mathsf{PV}}
\newcommand{\V}{\mathsf{V}}
\newcommand{\BASIC}{\mathsf{BASIC}}
\renewcommand{\S}{\mathsf{S}}

\newcommand{\Comp}{\mathit{C}}
\newcommand{\q}[1]{\textit{``#1''}}
\newcommand{\bt}{\mathit{bt}}
\newcommand{\poly}{\mathsf{poly}}

\newcommand{\commentout}[1]{}

\begin{document}

\ifdefined\JMLversion
\markboth{A. Atersias, S. Buss and M. M\"uller}{On the Consistency of Circuit Lower Bounds for Non-Deterministic Time}

%
\catchline{}{}{}{}{}
%

\title{On the Consistency of Circuit Lower Bounds\\ for Non-Deterministic Time\footnote{An extended abstract of part of this work appeared as \cite{ABM:NexpPpoly}.}}

\author{Albert Atserias}
\address{Universitat Polit\`ecnica de
    Catalunya \\ i Centre de Recerca Matem\`atica\\ Barcelona,
    Spain.\\ atserias@cs.upc.edu}

\author{Sam Buss}
\address{University of California, San Diego \\ La Jolla, California USA\\
sbuss@ucsd.edu}

\author{Moritz M\"uller}
\address{Universit\"at Passau \\ Passau, Germany \\ moritz.mueller@uni-passau.de}

\else 

\title{\bf On the Consistency of Circuit Lower Bounds\\ for Non-Deterministic Time\footnote{An extended abstract of part of this work appeared as \cite{ABM:NexpPpoly}.}}

\author{ Albert Atserias\footnote{Universitat Polit\`ecnica de
    Catalunya i Centre de Recerca Matem\`atica, Barcelona,
    Spain. Supported in part by Project PID2019-109137GB-C22 (PROOFS)
    and the Severo Ochoa and María de Maeztu Program for Centers and
    Units of Excellence in R\&D (CEX2020-001084-M) of the Spanish State
    Research Agency.}
  \and Sam Buss\footnote{University of California, San Diego, USA.
    Supported in part by Simons Foundation grant 578919.
  }\\
 \and Moritz M\"uller\footnote{Universit\"at Passau, Passau, Germany.}\\
}

\fi

\maketitle

\ifdefined\JMLversion
\begin{history}
\received{Day Month Year}
\revised{Day Month Year}
\accepted{Day Month Year}
\end{history}
\fi 

\begin{abstract}
  We prove the first unconditional consistency
  result for superpolynomial circuit lower bounds with a relatively
  strong theory of bounded arithmetic. Namely, we show that the
  theory~$\V^0_2$ is consistent with the conjecture
  that~$\NEXP\not\subseteq\Ppoly$, i.e., some problem that is solvable
  in non-deterministic exponential time does not have polynomial size
  circuits.  We suggest this is the best currently available evidence
  for the truth of the conjecture. The same techniques establish the
  same results with~$\NEXP$ replaced by the class of problems
  decidable in non-deterministic barely superpolynomial time such
  as~$\NTIME(n^{O(\log\log\log n)})$. Additionally, we establish a
  magnification result on the hardness of proving circuit lower
  bounds.
\end{abstract}

\thispagestyle{empty}

\newpage

\ifdefined\JMLversion
\else
\setcounter{page}{1}
\fi

\section{Introduction}
Bounded arithmetics are fragments of Peano arithmetic that formalize
reasoning with concepts and constructions of bounded computational
complexity. Their language is tailored so that natural classes of
bounded formulas define important complexity classes. For example, the
set of all bounded formulas defines precisely the problems in~$\PH$
and the set of~$\Sigma^b_1$-formulas those in~$\NP$.  The central
theories are comprised in Buss' hierarchy~\cite{Buss:bookBA}
\begin{equation}\label{eq:hierarchy}
\textstyle \S^1_2\subseteq\T^1_2\subseteq \S^2_2\subseteq \T^2_2\subseteq\cdots\subseteq\T_2\subseteq\V^0_2\subseteq \V^1_2
\end{equation}

The theory~$\S^1_2$ can be understood as
formalizing~$\PTIME$-reasoning, and~$\V^1_2$ as
formalizing~$\EXP$-reasoning. The levels of~$\T_2$ are determined by
induction schemes for properties of bounded computational
complexity. E.g.,~$\T^1_2$ has induction for~$\NP$, and~$\T_2$
for~$\PH$.  Intuitively, these theories can construct and reason with
polynomially large objects of various computational complexities.  The
theories~$\V^0_2$ and~$\V^1_2$ are extensions with a second sort of
variables ranging over bounded sets of numbers and are given by
comprehension schemes. Intuitively, these sets represent exponentially
large objects.

Low levels of the bounded arithmetic
hierarchy formalize a considerable part of contemporary complexity
theory.
This includes some advanced topics such
as the Arthur-Merlin hierarchy~\cite{Jerabek:approximatecounting}, hardness
amplification~\cite{Jerabek:thesis},
Toda's theorem~\cite{BKZ:collapsingmodular},
and
the~PCP~Theorem~\cite{Pich:PCP}. We refer to~\cite[Section~5]{MullerPich:SuccinctLB} for a
list of successful formalizations.  Concerning circuit complexity, the
topic of this paper, Je\v{r}\'{a}bek proved that his theory of
approximate counting~\cite{Jerabek:dualweakphp,Jerabek:thesis,Jerabek:approximatecounting}, which sits
below~$\T^2_2$, formalizes Rabin's primality test, and proves that it
is in~$\Ppoly$~\cite[Example~3.2.10, Lemma
3.2.9]{Jerabek:thesis}. Concerning lower bounds, many of the known
(weak) circuit lower bounds can be formalized in a theory of
approximate counting~\cite{MullerPich:SuccinctLB} and thus also in the
theory~$\T^2_2$.
For example, the~$\AC^0$ lower
bound for parity has been formalized in~\cite[Theorem~1.1]{MullerPich:SuccinctLB} via
probabilistic reasoning with Furst, Saxe and Sipser's random
restrictions~\cite{FSS:parity}, and in~\cite[Theorem~15.2.3]{Krajicek:book} via
Razborov's~\cite{Razborov:BAlowerbounds} proof of H\aa stad's switching lemma.

Razborov asked in his seminal work from 1995 for the ``right fragment
capturing the kind of techniques existing in Boolean complexity''
\cite[p.344]{Razborov:BAlowerbounds}.  Showing that any theory that is strong enough
to capture these techniques cannot prove lower bounds for general
circuits would give a precise sense in which current techniques are
insufficient. This however seems to be very difficult. We refer to
\cite[Introduction]{Razborov:PseudoHardkDnf} or \cite[Ch.27-30]{Krajicek:forcingbook} for a
description of the resulting research program, and to \cite{PichSanthanam:StrongCoNondet} for a
recent result.

In contrast to unprovability, the first and final words of
Kraj\'i\v{c}ek's 1995 monograph~\cite{Krajicek:book} ask for consistency
results\footnote{The citations to follow refer not to circuit lower
  bounds but to~$\PTIME \not= \NP$.},
namely to prove the conjecture in question ``for nonstandard models of
systems of bounded arithmetic''. These are ``not ridiculously
pathological structures, and a part of the difficulty in constructing
them stems exactly from the fact that it is hard to distinguish these
structures, by the studied properties, from natural numbers''
\cite[p.xii]{Krajicek:book}. In particular, showing that a given conjecture
is consistent with certain bounded arithmetics, already low ones,
would exhibit a world where both the conjecture and a considerable
part of complexity theory are true.

We therefore interpret consistency results as giving precise evidence for the {\em truth} of the conjecture. This is without doubt  preferable to appealing to intuitions, or alluding to the experience that the conjectures appear to be  theoretically coherent, exactly because a consistency result gives a precise meaning to this coherence.

\subsection{Previous consistency results}

Being well motivated, consistency results are also hard to come by,
and not much is known. In particular, it is unknown
whether~$\NP\not\subseteq\Ppoly$ is consistent with~$\S^1_2$.

It is not straightforward to formalize~$\NP\not\subseteq\Ppoly$
because exponentiation is not provably total in bounded
arithmetics. On the formal level, call a number~$n$ {\em small}
if~$2^n$ exists. A size-$n^c$ circuit can be coded by a binary string
of length at most~$10\cdot n^c\cdot \log(n^c)$, and hence by a number
below~$2^{10\cdot n^c\cdot \log(n^c)}$; this bound exists for
small~$n$.

On the formal level, an~$\NP$-problem is represented by
a~$\Sigma^b_1$-formula~$\varphi(x)$.  A sentence expressing that the
problem defined by~$\varphi(x)$ has size~$n^c$ circuits looks as
follows:
\begin{equation*}
\alpha^c_{\varphi}:=\ \forall n{\in}\Log\ \exists C{<}2^{n^c}\ \forall x{<}2^n\ (C(x){=}1 \leftrightarrow \varphi(x)).
\end{equation*}
Here, the quantifier on~$n$ ranges over small numbers above~$1$. We
think of the quantifier on~$C$ as ranging over circuits of
encoding-size~$n^c$, and of the quantifier on~$x$ as ranging over
length~$n$ binary strings.  Counting the~$\exists$ hidden
in~$\varphi$, this is a
bounded~$\forall\exists\forall\exists$-sentence (namely
a~$\forall\Sigma^b_3$-sentence).

Now more precisely, the central question whether~$\S^1_2$ is
consistent with~$\NP\not\subseteq\Ppoly$ asks for
a~$\Sigma^b_1$-formula~$\varphi(x)$ such
that~$\S^1_2+\big\{\neg \alpha^c_\varphi\mid c\in\N \big\}$ is
consistent. As mentioned a model witnessing this consistency would be
a world where a considerable part of complexity theory is true and
the~$\NP$-problem defined by~$\varphi$ does not have polynomial-size
circuits.  This is faithful in that there also exists
an~$\NP$-machine~$M$ that cannot be simulated by small circuits in the
model. Namely,~$\S^1_2$ proves that~$\varphi(x)$ is equivalent to a
formula
\begin{equation}\label{eq:NPmachine}
\exists y{<}2^{n^d} \q{$y$ is an accepting computation of $M$ on $x$}
\end{equation}
for a suitable~$\NP$-machine~$M$, namely a {\em model-checker}
for~$\varphi$. Here, the constant~$d$ stems from the polynomial
running time of~$M$. We write~$\alpha^c_M:=\alpha^c_\varphi$
for~$\varphi(x)$ equal to \eqref{eq:NPmachine}. One can also fix the
machine~$M$ in advance to a {\em universal} one, namely a
model-checker~$M^*$ for an~$\S^1_2$-provably~$\NP$-complete problem
(e.g.,~$\mathsf{SAT}$).

The predominant approach to the consistency of circuit lower bounds is
based on witnessing theorems: a proof of~$\alpha^c_M$ in some bounded
arithmetic implies a low-complexity algorithm that computes a
witness~$C$ from~$1^n$. E.g., if the theory has feasible witnessing
in~$\P$, then it does not prove~$\alpha^c_\varphi$ for any~$c$ unless
the problem defined by~$\varphi(x)$ is in~$\P$.  However,~$\S^1_2$ is
only known to have feasible witnessing in~$\P$ for
bounded~$\forall\exists$-sentences and~$\alpha^c_\varphi$ is
a~$\forall\exists\forall\exists$-sentence.

Fortunately, a self-reducibility argument implies that the quantifier
complexity of this formula can be reduced. Up to suitable changes
of~$c$, the formula~$\alpha^c_{M^*}$ is~$\S^1_2$-provably equivalent
to the following sentence of lower quantifier complexity:
\begin{equation*}
\begin{array}{lcl}
\beta^c_{M^*} &:=&
\forall n{\in}\Log\ \exists C{<}2^{n^c}\ \exists D{<}2^{n^c}\ \forall x{<}2^{n}\ \forall y{<}2^{n^d} \\
&&\quad (C(x){=}0 \to \neg \q{$y$ is an accepting computation of $M^*$ on $x$}) \ \wedge \\
&&\quad (C(x){=}1 \to \q{$D(x)$ is an accepting computation of $M^*$ on $x$}),
\end{array}
\end{equation*}
where $d$ stems from the polynomial runtime of $M^*$. We define
\begin{equation*}
\q{$\NP\not\subseteq\Ppoly$}\ :=\ \big\{\neg \beta^c_{M^*}\mid
c\in\N\big\}.
\end{equation*}

Note,~$\beta^c_{M^*}$ is a bounded~$\forall\exists\forall$-sentence
(namely a~$\forall\Sigma^b_2$-sentence). For such sentences,~$\S^2_2$
has feasible witnessing in~$\P^{\NP}$ \cite{Buss:bookBA}, and~$\S^1_2$
has feasible witnessing by certain interactive polynomial-time
computations \cite{Krajicek:nocounterexample}. This was exploited by Cook and
Kraj\'i\v{c}ek \cite{CookKrajicek:NPsubsetPpoly} to prove\footnote{$\P^{\NP}_{\mathrm{tt}}$
  denotes polynomial time with non-adaptive queries to
  an~$\NP$-oracle. In \cite{CookKrajicek:NPsubsetPpoly} a distinct but similar
  formalization of~$\NP\not\subseteq\Ppoly$ is used.}
that~\q{$\NP\not\subseteq\Ppoly$} is consistent with~$\S^2_2$
unless~$\PH\subseteq\P^{\NP}$, and with~$\S^1_2$
unless~$\PH \subseteq\P^{\NP}_{\mathrm{tt}}$.
%
%
Since the complexity of witnessing increases with the strength of the
theory, it seems questionable whether this method yields insights for
much stronger theories: by the Karp-Lipton
Theorem~\cite{KarpLipton:advice},~$\PH\not\subseteq\NP^{\NP}$ implies that
\q{$\NP\not\subseteq\Ppoly$} is true, and true sentences are
consistent with any true theory. Moreover, the focus of this work is
on unconditional consistency results.

Using similar methods, a recent line of works~\cite{KrajicekOliveira:Unprovability,BKO:Consistency,BydzovskyMuller:Ultrapowers,CKKO:LEARN}
achieved unconditional consistency results for fixed-polynomial lower
bounds, even for~$\P$ instead of~$\NP$ (based on~\cite{SanthanamWilliams:Uniformity}).  For
example, the main result in~\cite{BKO:Consistency} implies
that~$\S^2_2+\neg \alpha_{\varphi}^c$
and~$\S^1_2+\neg \alpha_{\psi}^c$ are consistent for certain
formulas~$\varphi(x)$ and~$\psi(x)$ that define problems in~$\P^\NP$
and~$\NP$, respectively.  Again it seems questionable whether the
underlying methods can yield insights for much stronger theories: by
Kannan \cite{Kannan:CircuitSize}, the lower bound stated by~$\neg \alpha_{\chi}^c$
is true for some formula~$\chi(x)$ defining a problem
in~$\NP^\NP$. Moreover, the formulas above depend on~$c$ and new ideas
seem to be required to reach the unconditional consistency of
superpolynomial lower bounds.

\subsection{New consistency results}\label{sec:newcons}

The purpose of this paper is to prove the unconditional consistency
of~$\NEXP\not\subseteq\Ppoly$ with the comparatively strong
theory~$\V^0_2$.  Consistency results for~$\V^0_2$ are meaningful,
since~$\V^0_2$ is stronger than~$\T^2_2$ which, as discussed earlier,
can formalize many results in complexity theory.  Our approach is not
via witnessing but via {\em simulating comprehension}.

The problems in~$\NEXP$ are naturally represented on the formal level
by $\hat\Sigma^{1,b}_1$-formulas~$\varphi(x)$: an existentially
quantified set variable followed by a bounded formula.
We discuss three ways to formalize $\NEXP\not\subseteq\Ppoly$, namely with
$\{\neg\alpha^c_{\varphi}\mid c\ge 1\}$ for a
$\hat\Sigma^{1,b}_1$-formula~$\varphi(x)$, with $\{\neg\alpha^c_{M_0}\mid c\ge 1\}$ and with
$\{\neg\beta^c_{M_0}\mid c\ge 1\}$
for a suitable universal $\NEXP$-machine~$M_0$. We now discuss these formalizations;
they are analogous to the formalizations discussed in the previous section.

 The ``direct formalization''
of the consistency of~$\NEXP\not\subseteq\Ppoly$ is based on the
formulas~$\alpha^c_\varphi$.  These are defined similarly as before
but with $\varphi$ a $\hat\Sigma^{1,b}_1$-formula:
\begin{definition}\label{def:alpha}
  Let~$c \in \N$ and let $\varphi = \varphi(x)$ be
  a $\hat\Sigma^{1,b}_1$-formula (with only one free
  variable~$x$, and in particular without free variables of the set
  sort). Define
\begin{equation*}
\alpha_\varphi^c\ :=\ \forall n{\in}\Log\ \exists C{\le} 2^{n^c}
  \forall x{<}2^n\ \big( C(x) \leftrightarrow \varphi(x)\big).
\end{equation*}
\end{definition}

Then our {\em direct formalization} of the consistency
of~$\NEXP\not\subseteq\Ppoly$ is:

\begin{theorem}\label{thm:alphaNEXP} There exists $\varphi(x)\in \hat\Sigma^{1,b}_1$
such that $\V^0_2+\big\{\neg \alpha^c_\varphi\mid c\in \N\big\}$
is consistent.
\end{theorem}

Theorem~\ref{thm:alphaNEXP} can be strengthened to establish the
consistency of $\NEXP\not\subseteq\PHpoly$ (see Section~\ref{sec:phpoly})
but our focus is on $\Ppoly$.

Theorem~\ref{thm:alphaNEXP} is proved in Section~\ref{sec:DirectConsisProof}
but in hindsight is not hard to prove. For~$\varphi(x)$ take a
formula negating the pigeonhole principle: it states that there exists
a set coding an injection from~$\{0,\ldots,x+1\}$
into~$\{0,\ldots,x\}$, and thus is expressible as a $\hat\Sigma^{1,b}_1$-formula.
The intermediate steps in the usual proof of the pigeonhole principle
involve further sets encoding injections, and these can also
expressed with $\hat\Sigma^{1,b}_1$-formulas.
If these formulas were computed by polynomial-size circuits,
then we could use quantifier-free induction to show that the
pigeonhole principle is provable in~$\V^0_2$. But it is well known
that this is not the case (see \cite[Corollary~12.5.5]{Krajicek:book}).

Concerning the faithfulness of the direct formalization we get, as
before, a model of~$\V^0_2$ where a certain~$\NEXP$-machine cannot be
simulated by small circuits. Indeed, for an {\em explicit}~$\NEXP$-machine~$M$ we can write
the formula \eqref{eq:NPmachine} using instead of~$\exists y$ a
quantification~$\exists Y$ for a set variable~$Y$:
\begin{equation}\label{eq:comp2}
\exists Y\q{$Y$ is an accepting computation of $M$ on $x$}.
\end{equation}
Roughly, an {\em explicit} $\NEXP$-machine is one such that
$\S^1_2$ can verify a suitable bound on its runtime;
we defer the details to Section~\ref{sec:explicit}.
It turns out that $\V^0_2$ proves that every
$\hat\Sigma^{1,b}_1$-formula~$\varphi(x)$ is
equivalent to \eqref{eq:comp2} for a suitable~$M$, namely a model-checker
for~$\varphi(x)$. Proving this is not trivial because $\V^0_2$ is
agnostic about the existence of computations of exponential-time
machines. One of our contributions is to prove it; we
give the details in Section~\ref{sec:mc}.
\begin{definition}\label{def:alphaM}
For an explicit $\NEXP$-machine $M$ and $c\in\N$ we set  $\alpha^c_M:=\alpha^c_\psi$ where $\psi$ is the formula~\eqref{eq:comp2}.
\end{definition}


Intuitively,~$\V^0_2$ does not know whether non-trivial
exponential-size sets exist, namely sets not given by bounded
formulas.  But then, how meaningful is the consistency
statement of Theorem~\ref{thm:alphaNEXP} or the corresponding
statement for $\{\neg\alpha^c_M\mid c\ge 1\}$?
These sentences contain (universal and) existential set quantifiers.
It turns out that we
can move again to a suitably modified sentence~$\beta^c_M$ of lower
quantifier complexity, namely a sentence all of whose set quantifiers are universal (i.e.,~$\forall\Pi^{1,b}_1$): such sentences do not entail the
existence of non-trivial large sets. This does not follow from simple
self-reducibility arguments but is a deep result of complexity theory,
namely the Easy Witness Lemma of Impagliazzo, Kabanets and Wigderson
\cite[Theorem~31]{IKW:EasyWitness}.  We use Williams' version as stated
in~\cite[Lemma~3.1]{Williams:Improving} (see
\cite[Theorem~3.1]{Williams:NaturalVersus} for the equivalence):

\begin{lemma}[Easy Witness Lemma] \label{lem:ewl}
  If~$\NEXP\subseteq\Ppoly$, then every~$\NEXP$-machine has
  polynomial-size oblivious witness circuits.
\end{lemma}

An {\em oblivious witness circuit} for a machine~$M$ and input
length~$n$ is a circuit~$D$ with at least~$n$ inputs such that for
every~$x$ of length~$n$, if~$M$ accepts~$x$, then~$\mathit{tt}(D_x)$
encodes an accepting computation of~$M$ on~$x$. Here, the
circuit~$D_x$ is obtained from~$D$ by fixing the first~$n$ inputs to
the bits of~$x$, and~$\mathit{tt}(D_x)$ is the truth table of~$D_x$.
In the statement of the lemma, {\em polynomial-size} refers to
polynomial in~$n$, and the qualifier {\em oblivious} refers to the
fact that~$D$ depends only on the length of~$x$, not on~$x$ itself.

In the language of two-sorted bounded arithmetic the
string~$\mathit{tt}(D_x)$ corresponds to the set~$D_x(\cdot)$ of
numbers accepted by~$D_x$. We thus define the formula $\beta^c_M$
by replacing~$D(x)$ by~$D_x(\cdot)$ and~$\forall y$
by~$\forall Y$:

\begin{definition}\label{def:beta} For $c\in\N$ and an explicit $\NEXP$-machine $M$ we set
\begin{equation*}
\begin{array}{lcl}
\beta_M^c&:=&\forall n{\in}\Log\  \exists C{<}2^{n^c}\ \exists D{<}2^{n^c}\
\forall x{<}2^n\ \forall Y \\
&&\quad (C(x){=}0\ \to\ \neg\q{$Y$ is an accepting
computation of $M$ on $x$})\ \wedge\\
&&\quad (C(x){=}1\ \to\ \q{$D_{x}(\cdot)$ is an accepting computation
of $M$ on $x$}).
\end{array}
\end{equation*}
\end{definition}
In Section~\ref{sec:universal} we define
a suitable universal explicit~$\NEXP$-machine~$M_0$ and arrive at our formalization of $\NEXP\not\subseteq\Ppoly$:
\begin{definition}\label{def:NexpNotPpoly}
$
\q{$\NEXP\not\subseteq\Ppoly$}:= \{\neg\beta^c_{M_0}\mid
c\in\N\}.
$\end{definition}

The main result of this paper is:
\begin{theorem}\label{thm:NEXP}
The theory $\V^0_2$ is consistent with both formalizations of
$\NEXP\not\subseteq\Ppoly$; concretely,
$\V^0_2 + \{ \lnot\alpha^c_{M_0} : c\in \N \}$ and
$\V^0_2 + \{ \lnot\beta^c_{M_0} : c\in \N \}$ are consistent.
\end{theorem}

In the notation introduced above, this gives:

\begin{corollary}\label{coro:NEXPsuccinct}
$\V^0_2+\q{$\NEXP\not\subseteq\Ppoly$}$ is consistent.
\end{corollary}

Both $\{ \lnot\alpha^c_{M_0} : c\in \N \}$ and $\{ \lnot\beta^c_{M_0} : c\in \N \}$ are formalizations of
$\NEXP\not\subseteq\Ppoly$. The first has the advantage of being more direct whereas the second has the advantage of having lower quantifier complexity: $\beta^c_{M_0}$ is $\forall\Pi^{1,b}_1$ while $\alpha^c_{M_0}$  is $\forall\Sigma^b_\infty(\Pi^{1,b}_1)$. In addition, being $\forall\Pi^{1,b}_1$ is instrumental for our magnification result discussed below (Theorem~\ref{thm:mag}).
%
It is easy to see that $\V^0_2$ proves that $\{ \lnot\alpha^c_{M_0} : c\in \N \}$
implies $\{ \lnot\beta^c_{M_0} : c\in \N \}$. The converse implication
 is true too, but depends on the
Easy Witness Lemma.
It is open whether $\V^0_2$ proves this implication or the Easy Witness Lemma.

We emphasize here that our formalization
of~$\NEXP \not\subseteq \Ppoly$ through the universal machine~$M_0$
and the $\alpha^c_{M_0}$ and~$\beta^c_{M_0}$ sentences
refers exclusively to the setting of
non-relativized complexity classes.

\medskip

Second we show that~$\NEXP$ can be lowered to just above~$\NP$.
For~$k \in \N$, define~$\log^{(k)}n$ inductively
by~$\log^{(1)}n:= \log n$, and~$\log^{(k+1)}n:=\log\log^{(k)} n$. We
prove:

\begin{theorem}\label{thm:NTIME}
$\V^0_2+\q{$\NTIME(n^{O(\log^{(k)}n)})\not\subseteq\Ppoly$}$ is
  consistent for every positive $k \in\penalty10000 \N$.
\end{theorem}

The formalization and proof proceeds similarly and relies on an Easy
Witness Lemma for barely superpolynomial time by Murray and
Williams~\cite{MurrayWilliams:CircuitLB}. Theorem~\ref{thm:NTIME} ``almost'' settles the
central question for the consistency of~$\NP\not\subseteq\Ppoly$ with
a strong bounded arithmetic. Closing the tiny gap, however, seems to
require some new ideas.

\subsection{Simulating comprehension}\label{sec:introcomprehension}

The proof of the consistency of circuit lower
bounds is based on the complexity of constant depth
propositional proofs for the pigeonhole principle. We shall see that
$\V^0_2 + \alpha^c_{M_0}$ (and thus $\V^0_2 + \beta^c_{M_0}$) proves the
pigeonhole principle. This implies
Theorem~\ref{thm:NEXP} as it is well-known that~$\V^0_2$ cannot prove
this principle. Thereby, Theorem~\ref{thm:NEXP} is ultimately based on
the exponential lower bound for
this principle  in
bounded depth Frege systems~\cite{Ajtai:PHP,BIKPPW:PHP}. On a high
level, while the approach based on witnessing uses complexity
theoretic methods, our approach is based on methods that arose from mathematical
logic, in particular forcing (cf.~\cite{AtersiasMuller:forcingBA}).

The $\{\lnot \beta^c_{M_0}\}$ formulation of $\q{$\NEXP\not\subseteq\Ppoly$}$
provides an additional insight into the consistency lower bound.
By the Easy Witness Lemma, the
inclusion $\NEXP\subseteq\Ppoly$ implies that a rich collection of
sets is represented by circuits (via their truth tables).  A weak
theory can quantify over circuits and hence implicitly over this
collection. Thus, intuitively, $\beta^c_{M_0}$~should enable a weak
theory to simulate a two-sorted theory of considerable strength.
More precisely, we show that $\beta^c_{M_0}$ can be used to simulate a
considerable fragment of $\Sigma^{1,b}_1$-comprehension, i.e., a
considerable fragment of~$\V^1_2$.

The sketched idea can be made explicit as follows. By~$\S^1_2(\alpha)$
we denote the two-sorted variant of~$\S^1_2$. Its models consist of
two universes~$M$ and~$\mathcal{X}$ interpreting the number and the
set sort, respectively. Given such a model that additionally
satisfies~$\beta^c_{M_0}$ for some~$c\in\N$,
we will show in Lemma~\ref{lem:V12} that
shrinking~$\mathcal{X}$ to the sets represented by circuits in~$M$
yields a model of~$\V^1_2$.  This has two interesting consequences.
The first is:

\vbox{  
\begin{theorem} \label{thm:conservative}
Let~$\T$ be a theory that contains~$\S^1_2(\alpha)$ but does not
prove all number-sort consequences
of~$\V^1_2$. Then~$\T+\q{$\NEXP\not\subseteq\Ppoly$}$ is consistent.
\end{theorem}
}   

By a {\em number-sort} formula we mean one that does not use set-sort
variables.  Note that the corollary refers to number-sort sentences of
arbitrary unbounded quantifier complexity.  It is conjectured
that~$\V^1_2$ has more number-sort consequences than all other
theories mentioned so far. But this is known only for~$\S^1_2$
\cite{Takeuti:BATruthDef,Krajicek:exponentiation}, and there even for~$\forall\Pi^b_1$-sentences.
Theorem~\ref{thm:conservative} directly infers evidence for the
truth of \q{$\NEXP\not\subseteq\Ppoly$} from progress in mathematical
logic on understanding independence.  Loosely speaking, we view it in
line with the belief that it is mathematical logic that ultimately
bears on fundamental complexity-theoretic conjectures (see e.g.\ again
the preface of \cite{Krajicek:book}).

The second consequence is:

\begin{theorem}\label{thm:mag}
  If~$\S^1_2(\alpha)$ does not prove \q{$\NEXP\not\subseteq\Ppoly$},
  then~$\V^1_2$ does not prove \q{$\NEXP\not\subseteq\Ppoly$}.
\end{theorem}

This is a {\em magnification result} on the hardness of proving
circuit lower bounds: it infers strong hardness (for~$\V^1_2$) from
weak hardness (for~$\S^1_2(\alpha)$). The term magnification has been
coined in \cite{OliveiraSanthanam:Magnification} in the context of circuit lower bounds where such
results are currently intensively investigated (cf.\
\cite{CHOPRS:BeyondNatural}). In proof complexity such results are rare so far. An
example in propositional proof complexity appears in
\cite[Proposition~4.14]{MullerPich:SuccinctLB}.
Magnification results are interesting because they reveal
inconsistencies in common beliefs about what is and what is not within
the reach of currently available techniques.
Theorem~\ref{thm:mag} might
foster hopes to complete
Razborov's program to find a precise barrier in circuit complexity (cf.~Remark~\ref{rem:hope}).

\section{Consistency of the direct formalization } \label{sec:direct}

In this section we provide the details of the simple proof of
Theorem~\ref{thm:alphaNEXP}. We begin by recalling the necessary
preliminaries on bounded arithmetic. This will be needed also in later
sections. We refer to \cite[Ch.5]{Krajicek:book} for the missing details.

\subsection{Preliminaries: bounded arithmetic}

Bounded arithmetics have
language $x{\le}y$, $0$, $1$, $x{+}y$, $x{\cdot}y$, $\floor{x{/}2}$, $x{\#}y$, $|x|$,
and built-in equality~$x{=}y$.  Note that Cantor's pairing $\langle
x,y\rangle$ is given by a term. Iterating it gives $\langle
x_1,\ldots,x_k\rangle$ for $k>2$. A number~$x$ is called \emph{small}
if it satisfies the formula $\exists y\ x{=}|y|$. We
abbreviate $\exists y\ x{=}|y|$ by $x{\in}\Logold$
and $x{\in}\Logold \wedge 1{<}x$ by $x{\in}\Logold_{>1}$.  The
quantifiers $\forall x{\in}\Logold_{>1}$
and $\exists x{\in}\Logold_{>1}$ range over small numbers above~$1$.
If~$x = |y|$, we write~$2^x$ for~$1\#y$ and similarly for other
exponential functions.  E.g., a formula of the
form $\forall x{\in} \Log\ \ldots\ 2^{x^2} \ldots$ stands for the
formula $\forall x \forall y\ (1{<}x\wedge x{=}|y|\to \ldots\ y\#y
\ldots)$.

\paragraph{Theories.}
The theories of bounded arithmetic are given by a set~$\BASIC$ of
universal sentences determining the meaning of the symbols, plus
induction schemes. For a set of formulas~$\Phi$, the set (of the
universal closures) of formulas
$$
\varphi(\bar x,0)\wedge\forall y{<}z\ (\varphi(\bar x,y)\to\varphi(\bar
x,y+1))\to \varphi(\bar x,z),
$$
for~$\varphi\in\Phi$, is the scheme
of~\emph{$\Phi$-induction}. Restricting to small numbers~$z$ gives the
scheme of~\emph{$\Phi$-length induction}; formally, replace~$z$
by~$|z|$ above. Here, and throughout, when writing a formula~$\psi$
as~$\psi(\bar x)$ we mean that {\em all} free variables of~$\psi$ are
among~$\bar x$.

The set~$\Sigma^b_\infty$ contains all bounded formulas,
and~$\Sigma^b_i,\Pi^b_i$, for~$i\in\N$, are subsets
of~$\Sigma^{b}_\infty$ that are defined by counting alternations of bounded
quantifiers~$\exists x{\le }t, \forall x{\le}t$, not counting sharply
bounded ones~$\exists x{\le }|t|, \forall x{\le}|t|$.  In
particular,~$\Sigma^b_0=\Pi^b_0$ is the set of sharply bounded
formulas.  The theories~$\T^i_2$ are defined
by~$\BASIC + \text{$\Sigma^b_i$-induction}$.  The theories~$\S^i_2$
are defined by~$\BASIC + \text{$\Sigma^b_i$-length-induction}$. Full
bounded arithmetic~$\T_2:=\bigcup_{i\in\N}\T^i_2$
has~$\Sigma^{b}_\infty$-induction.


\paragraph{Two-sorted theories.}
Two-sorted bounded arithmetics are obtained by adding a new set of
variables~$X,Y,\ldots$ of the {\em set sort}. Original
variables~$x,y,\ldots$ are of the {\em number sort}.  We shall use
capital letters also for number-sort variables. Therefore, for
clarity, from now on we write~$\exists_2 X$ and~$\forall_2X$ for
quantifiers on set-sort variables~$X$.  The language is enlarged by
adding a binary relation~$x{\in}X$ between the number and the set
sort.  A {\em number-sort} formula is one that uses only the number
sort. In particular, it has no set-sort parameters. By a {\em term} we
mean a term in the number sort.  We write~$X{\le}z$ for~$\forall
y\ (y{\in}X\to y{\le}z)$.

Models have the form~$(M,\mathcal X)$ where~$M$ is a universe for the
number sort and~$\mathcal X$ is a universe for the set sort. The
symbol~$\in$ is interpreted by a subset of~$M\times\mathcal X$. The
standard model is~$(\N,[\N]^{<\omega})$ where~$[\N]^{<\omega}$ is the
set of finite subsets of~$\N$; the number sort symbols are interpreted
as usual over~$\N$ and~$\in$ by actual element-hood.

The sets~$\Sigma^b_\infty(\alpha),\Sigma^b_i(\alpha),\Pi^b_i(\alpha)$
are defined as~$\Sigma^b_\infty,\Sigma^b_i,\Pi^b_i$, allowing free
set-variables and the symbol~$\in$, but not allowing set-sort
quantifiers, nor set-sort equalities~$X{=}Y$.  Another name for the
set~$\Sigma^b_{\infty}(\alpha)$ is~$\Sigma^{1,b}_0$.
The
theories~$\T^i_2(\alpha)$,~$\S^i_2(\alpha)$, and~$\T_2(\alpha)$, are
given by~$\BASIC$ and analogous induction schemes as before,
namely~$\Sigma^b_i(\alpha)$-induction,~$\Sigma^b_i(\alpha)$-length
induction, and~$\Sigma^b_\infty(\alpha)$-induction,
respectively. Additionally, we add the following axioms with the set
sort. Recalling the notation~$X{\leq}z$ introduced above, the new
axioms are~(the universal closures of):

\medskip
\begin{tabular}{lll}
{\em set-boundedness axiom}: & $\exists z \ X{\le}z$. \\
{\em extensionality axiom}: & $X{\le}z\wedge Y{\le}z\wedge \forall y{\le} z\ (y{\in}X\leftrightarrow y{\in}Y)\to X{=}Y$.
\end{tabular}
\medskip

\noindent We add the scheme of (bounded)
\emph{$\Delta^b_1(\alpha)$-comprehension}, given by (the universal
closures of) the formulas
\begin{equation}
 \exists_2 Y{\le} z\ \forall y{\le}z\ \big(y\in
  Y\leftrightarrow\varphi(\bar X,\bar x,y)\big),
\label{eq:comprehensionscheme}
\end{equation}
where~$\varphi(\bar X,\bar x,y)$ is~$\Delta^b_1(\alpha)$ with respect
to~the theory defined over the two-sorted language as~$\BASIC$
plus~$\Sigma^b_1(\alpha)$-length-induction, i.e., this theory
proves~$\varphi(\bar X,\bar x,y)$ equivalent to both
a~$\Pi^b_1(\alpha)$-formula and a~$\Sigma^b_1(\alpha)$-formula.

For example, this scheme implies that there is a set~$Y$ as described
when $\varphi(\bar X,\bar x,y)$ is $f^{\bar X}(\bar x,y){=}1$
where $f^{\bar X}(\bar x,y)$ is a function that
is $\Sigma^b_1(\alpha)$-definable in~$\S^1_2(\alpha)$. The superscript
indicates that $\bar X$ comprises all the free variables of the set
sort that appear in the $\Sigma^b_1(\alpha)$-formula that
defines $f^{\bar X}(\bar x,y)$.  It is well known~\cite{Buss:bookBA} that
these are precisely the functions that are computable in polynomial
time with oracles denoted by the set variables. We do not
distinguish~$\S^1_2$ (or~$\S^1_2(\alpha)$) from its variant in the
language~$\PV$ (resp.,~$\PV(\alpha)$) which has a symbol for all
polynomial time functions (resp., with oracles denoted by the set
variables). We shall often use that $\S^1_2(\alpha)$ proves induction
for quantifier-free $\PV(\alpha)$-formulas
(cf. \cite[Lem\-ma~5.2.9]{Krajicek:book}). We write
quantifier-free $\PV(\alpha)$-formulas with latin capital letters;
e.g., $F(\bar X,\bar x)$.

\paragraph{A piece of notation.}
For formulas~$\varphi(Y,\bar X,\bar x)$ and~$\psi(\bar Z,\bar z,u)$ we
write
$$
\varphi\big( \psi(\bar Z,\bar z,\cdot),\bar X,\bar x \big)
$$
for the formula obtained from~$\varphi$ by replacing every atomic
subformula of the form~$t{\in}Y$, for~$t$ a term, by the
formula~$\psi(\bar Z,\bar z,t)$, preceded by any necessary renaming
of the bound variables of~$\varphi$ to avoid the capturing of free
variables.  We use this notation only for formulas~$\varphi$ without
set equalities.

\paragraph{Genuine two-sorted theories.}
It is easy to see that the theories $\T^i_2(\alpha),\S^i_2(\alpha)$
have the same number sort consequences as~$\T^i_2,\S^i_2,$
respectively. Also $\T^i_2(\alpha),\S^i_2(\alpha)$ are conservative
over their subtheories without $\Delta^b_1(\alpha)$-comprehension.
Intuitively, the two-sorted versions of bounded
arithmetics are the usual ones plus syntactic sugar.  Genuine
set-sorted theories are obtained from~$\T_2(\alpha)$ by adding
\emph{(bounded) $\Phi$-comprehension} for certain sets of 
formulas~$\Phi$, i.e.,~\eqref{eq:comprehensionscheme}
for $\varphi(\bar X,\bar x,y)$ in~$\Phi$.

The set~$\Sigma^{1,b}_\infty$ contains all two-sorted formulas with
quantifiers of both sorts, but bounded number-sort quantifiers. Again
we disallow set equalities.
The
sets~$\Sigma^{1,b}_i,\Pi^{1,b}_i$, for~$i\in\N$, are subsets
of~$\Sigma^{1,b}_\infty$ defined by counting the alternations of
set quantifiers (and not counting number quantifiers).
A~$\hat\Sigma^{1,b}_1$-formula is of the form
\begin{equation}
\exists_2 Y\ \varphi(\bar X,Y,\bar x)
\label{eqn:generichatSigma}
\end{equation}
where~$\varphi(\bar X,Y,\bar x)$ is a $\Sigma^{1,b}_0$-formula.

For~$i\in\N$ the theory~$\V^i_2$ is given
by $\Sigma^{1,b}_i$-comprehension. In particular, $\V^0_2$ is given
by $\Sigma^{1,b}_0$-comprehension. It has the same
number-sort consequences as~$\T_2$.

\begin{remark}\label{rem:setbound}
  Sometimes, the sets~$\Sigma^{1,b}_i(\alpha)$ are defined with
  bounded set quantifiers $\exists X{\le t}$ and~$\forall
  X{\le}t$. The difference is not essential: for
  every $\Sigma^{1,b}_\infty$-formula~$\varphi(\bar X,Y,\bar x)$ there
  is a term~$t(\bar x)$ such that $\S^1_2(\alpha)$ proves
\begin{equation*}
t(\bar x){\le}y\to \big( \varphi(\bar X,Y,\bar
x)\leftrightarrow\varphi(\bar X,Y^{\le y},\bar x)\big)
\end{equation*}
where $Y^{\le y}$ stands for $\psi(Y,y,\cdot)$
with $\psi(Y,y,u):=(u{\le}y\wedge
u{\in}Y)$. By $\Delta^b_1(\alpha)$-comprehen\-sion, $\exists_2Y\varphi$
is $\S^1_2(\alpha)$-provably equivalent to $\exists_2Y{\le}t(\bar
x)\ \varphi$. It follows that every $\Sigma^{1,b}_i(\alpha)$-formula
is $\S^1_2(\alpha)$-provably equivalent to one with bounded set sort
quantifiers.
\end{remark}

\begin{remark} Disallowing
set equalities is convenient but inessential in the sense
that $\V^i_2$ does not change when set equalities are allowed
in~$\Sigma^{1,b}_i$.  Indeed, let $\varphi(\bar X,\bar x)$ be
a $\Sigma^{1,b}_i$-formula except that set equalities are
allowed. Then there is a $\Sigma^{1,b}_i$-formula $\varphi^*(\bar
X,\bar x,u)$ (without set equalities and) with bounded set quantifiers
such that $\S^1_2(\alpha)$ proves
 $$
 \exists u\ \big( \varphi(\bar X,\bar x)\leftrightarrow\varphi^*(\bar X,\bar x,u) \big).
 $$
\end{remark}

\begin{proof} The formula $\varphi^*$  is defined by a
straightforward recursion on~$\varphi$. For example, if $\varphi$
is $X_1{=}X_2$, then $\varphi^*$ is~$\forall y{\le} u\ (y{\in}X_1\to
y{\in}X_2)\wedge \forall y{\le} u\ (y{\in}X_2\to y{\in}X_1)$; a~$u$
witnessing the equivalence is any common upper bound on $X_1$
and~$X_2$. If $\varphi$ is $\exists_2Y \psi(\bar X,Y,\bar x)$
and $\psi^*=\psi^*(\bar X,Y,\bar x,u)$ is already defined,
then $\varphi^*$ is $\exists_2Y{\le} t(\bar x,u)\ \psi^*(\bar X,Y,\bar
x, u)$ where the term~$t$ is chosen according to the previous remark.
\end{proof}

\paragraph{Circuits.}\label{sec:circ}
A circuit with~$s$ gates is coded by a number
below~$2^{10\cdot s\cdot|s|}$.  On the formal level we shall only
consider small circuits, i.e.,~$s\in \Logold$,
so~$2^{10\cdot s\cdot|s|}$ exists.  We use capital letters $C,D,E$ for
number variables when they are intended to range over circuits.  There
is a $\PV$-function~$\eval(C,x)$ that (in the standard model) takes a
circuit~$C$ with, say, $n\le |C|$ input gates, and evaluates it on
inputs~$x<2^{n}$.  This means that the input gates of~$C$ are assigned
the bits of the length-$n$ binary representation of~$x$; we
assume $\eval(C,x)=0$ if $x\ge 2^n$ or if $C$ does not code a circuit.

It is notationally convenient to have circuits take finite
tuples $\bar x=(x_1,\ldots, x_k)$ as inputs; formally, such a circuit
has~$k$ sequences of input gates, the $i$-th taking the bits
of~$x_i$. Again, $\eval(C,\bar x)$ denotes the evaluation function; it
outputs~$0$ if any~$x_i$ has length bigger than the length of its
allotted input sequence. Our circuits have exactly one output gate,
so~$\S^1_2$ proves $\eval(C,\bar x){<}2$.  We write~$C(\bar x)$ for
the quantifier-free $\PV$-formula $\eval(C,\bar x){=}1$; in some
places we also write $C(\bar x){=}1$ and $C(\bar x){=}0$ instead
of $C(\bar x)$ and $\neg C(\bar x)$, respectively.

For a circuit~$C$ taking~$(\ell+k)$-tuples as inputs and
an~$\ell$-tuple~$\bar x$ we let~$C_{\bar x}$ be the circuit obtained
by fixing the first~$\ell$ inputs to~$\bar x$; it takes~$k$-tuples as
inputs. Formally,~$C_{\bar x}$ is a~$\PV$-term with
variables~$C,\bar x$ and~$\S^1_2(\alpha)$
proves~$(C_{\bar x}(\bar y) \leftrightarrow C(\bar x,\bar y))$
and~$|C_{\bar x}|{\le}|C|$.

\begin{lemma}\label{lem:pvcircuit}
  For every quantifier-free $\PV$-formula~$F(\bar x)$ there is
  a~$c\in\N$ such that~$\S^1_2$ proves
$$
\forall n{\in} \Log\ \exists C{<}2^{n^c}\ \forall \bar x{<}2^n\ \big( C(\bar x)\leftrightarrow F(\bar x)\big).
$$
\end{lemma}

\noindent On the formal level, if~$Y$ is a set and~$C$ is a circuit,
then we say that~$Y$ is {\em represented} by~$C$ if~$\forall
y\ (C(y)\leftrightarrow y{\in}Y)$.  In our notation, such set~$Y$ is
written~$C(\cdot)$, or~$\eval(C,\cdot){=}1$. More precisely, for a
formula~$\varphi(Y,\bar X,\bar x)$ and a circuit~$C$ we write
$$
\varphi\big( C(\cdot),\bar X,\bar x \big),
$$
for the formula obtained from~$\varphi$ by replacing every formula of
the form $t{\in}Y$ by~$C(t)$, i.e., by $\eval(C,t){=}1$.  Note that if
the set~$Y$ is represented by a circuit with~$n$ inputs,
then $Y{<}2^n$, provably in~$\S^1_2$.  For example, we shall use
circuits to represent computations of exponential-time
machines~$M$. Using the notation introduced in Section
\ref{sec:explicitmachine},
$$
\q{$C(\cdot)$ is a halting computation of $M$ on $\bar x$}
$$
is a~$\Pi^b_1$-formula with free variables~$C,\bar x$ stating that the
circuit~$C$ represents a halting computation of~$M$ on~$\bar x$.

\subsection{Consistency of the direct formalization for $\NEXP$}\label{sec:DirectConsisProof}

The set of~$\hat\Sigma^{1,b}_1$-formulas without free variables of the
set sort is a natural class of formulas defining, in the standard
model, all the problems in~$\NEXP$. For such a formula~$\psi$ it is
straightforward to write down a set of sentences (a.k.a.\ a theory)
stating that~$\psi$ does not have polynomial-size circuits. We
explicitly define this direct formalization of~$\NEXP\not\subseteq\Ppoly$ as
the set of all sentences of the form~$\neg\alpha^c_{\psi}$, for $c \in
\N$, for the sentence~$\alpha^c_{\psi}$ defined in the introduction,
and then argue that its consistency with~$\V^0_2$ follows from known
lower bounds in proof complexity.

We are ready to prove Theorem~\ref{thm:alphaNEXP}.

\begin{proof}[Proof of Theorem~\ref{thm:alphaNEXP}:]
The {\em
    (functional) pigeonhole principle}~$\PHP(x)$ is
  the following~$\Pi^{1,b}_1$-formula:
\begin{align*}
\forall_2 X\ \big(&\exists y{\le}x{+}1\ \forall z{\le}x\ \neg\langle y,z\rangle{\in}X\ \vee \\
&  \exists y{\le} x{+}1\  \exists z{\le}x\ \exists z'{\le}x\ (\neg z{=}z'\wedge\langle y,z\rangle{\in}X\wedge \langle y,z'\rangle{\in}X)\ \vee \\
&  \exists y{\le}x{+}1\ \exists y'{\le} x{+}1\ \exists z{\le}x \ (\neg y{=}y'\wedge\langle y,z\rangle{\in}X\wedge \langle y',z\rangle{\in}X)
 \big).
\end{align*}
Note that~$\psi = \psi(x) := \neg\PHP(x)$ is (logically equivalent to)
a~$\hat\Sigma^{1,b}_1$-formula. For the sake of contradiction assume
that~$\V^0_2+\big\{\neg \alpha^c_{\psi}\mid c\in \N\big\}$ is
inconsistent. By compactness, there exists $c\in\N$ such that~$\V^0_2$
proves~$\alpha^c_{\psi}$.

\medskip

\noindent{\em Claim:} $\V^0_2+\alpha^c_{\psi}$ proves $\PHP(x)$.

\medskip

The claim implies the theorem: it is well known
\cite[Corollary~12.5.5]{Krajicek:book} that there is an expansion~$(M,R^M)$
of a model~$M$ of~$\BASIC$ by an interpretation~$R^M\subseteq M$ of a
new predicate~$R$ such that~$R^M$ is bounded and
witnesses~$\neg\PHP(n)$ for some~(nonstandard)~$n\in M$, and,
further,~$(M,R^M)$ models induction for bounded
formulas. Let~$\mathcal{Y}$ be the collection of bounded sets
definable in~$(M,R^M)$ by bounded formulas. Then~$(M,\mathcal{Y})$ is
a model of~$\V^0_2$ with~$R^M\in\mathcal Y$,
so~$(M,\mathcal Y)\models\neg\PHP(n)$.

\medskip

We are left to prove the claim. Argue in~$\V^0_2$ and
set~$n:=\max\{|x|,2\}$. Then~$\alpha^c_{\psi}$ gives a circuit~$C$
such that
\begin{equation*}
\forall u{\le}x\ (\neg C(u)\leftrightarrow\PHP(u)).
\end{equation*}

We observe that~$\V^0_2$ proves that~$\PHP(x)$ is inductive, i.e.,
\begin{equation}
\PHP(0)\wedge \forall u{<}x\ (\PHP(u)\to\PHP(u+1)). \label{eqn:phpind}
\end{equation}
Indeed, if~$X$ is a set that witnesses $\neg\PHP(u+1)$, then we construct a
set~$Y$ that witnesses $\neg\PHP(u)$ as follows. If there does not
exist any $v{\leq}u{+}1$ with $\langle v,u\rangle{\in}X$, then the
set~$Y := X$ itself is the witness we want. On the other hand, if
there exists $v{\leq}u{+}1$ with~ $\langle v,u\rangle{\in}X$, then
let $Y$ be the set of pairs $z = \langle x,y\rangle$ such that the two
projections $x = \pi_1(z)$ and~$y = \pi_2(z)$ satisfy the
formula $\varphi(x,y,u,v)$ below, for the fixed parameters $u$
and~$v$:
\begin{align*}
\varphi(x,y,u,v) := x{\le}u \wedge y{<}u\wedge\big((x{>}v \wedge \langle x{-}1,y\rangle{\in}X)
\vee (x{<}v \wedge \langle x,y\rangle{\in}X)\big).
\end{align*}
Here,~$x{-}1$ denotes the (truncated) predecessor~$\PV$-function.
In the definition of~$Y$ we used the two projections~$\pi_1$ and~$\pi_2$,
also as~$\PV$-functions. Since the definition of~$Y$ is a
quantifier-free~$\PV(\alpha)$-formula, the set~$Y$ exists by
quantifier-free~$\PV(\alpha)$-comprehension, and it is clear by
construction that it witnesses~$\neg\PHP(u)$.

To complete the proof, plug~$\neg C(u)$ for~$\PHP(u)$
in~\eqref{eqn:phpind} and quantifier-free~$\PV(\alpha)$-induction
gives~$\neg C(x)$, and hence~$\PHP(x)$.
\end{proof}

\begin{remark} The model~$(M,\mathcal{X})$ that witnesses the above
  consistency is a model of~$\V^0_2$ where~$\PHP(n)$ fails for some
  nonstandard~$n\in M$: otherwise~$\alpha^{1}_{\neg\PHP}$ would be
  true and witnessed by trivial circuits that always reject.
\end{remark}

\subsection{A strengthening to $\PHpoly$}\label{sec:phpoly}

While our focus is on $\Ppoly$, in this section we point out a version of
Theorem~\ref{thm:alphaNEXP} stating the consistency of $\NEXP\not\subseteq\PHpoly$.


For $i>0$, let $T_i(e,t,x)$~denote a universal
 $\Sigma^b_i$-formula: for every
$\Sigma^b_i$-formula $\varphi(x)$, there are
$e,d\in\N$ such that  $\V^0_2$ (in fact, $\S^1_2$ \cite[Corollary 6.1.4]{Krajicek:book})
proves
$$\varphi(x)\leftrightarrow T_i(e,2^{|x|^d+d},x).
$$
Intuitively, the parameter
$|x|^d+d$ serves as a runtime bound of a suitable model-checker coded by $e$. Thus, the
formulas $T_i(e,2^{|x|^d+d},x)$ for varying $c,d\in\N$ define (in the standard model) precisely the problems in the $i$-th level $\Sigma^\PTIME_i$of the polynomial hierarchy $\PH$.

We incorporate nonuniformity as follows. Again, let $\pi_1,\pi_2$ be the $\PV$-functions computing the projections for pairs $\langle x,y\rangle$. Define
$$T_i'(a,x):=T_i(\pi_1(a),2^{|a|},\langle\pi_2(a),x\rangle).
$$
Thus, $a$ determines the runtime bound and some ``advice'' $\pi_2(a)$.
Then $Q\subseteq\N$ is in $\PHpoly$
 if there exists $i>0$ and a function~$a(n)$ such
that $|a(n)|$ is polynomially bounded in~$n$ and such that for all $x$ we have
$x\in Q$ if and only if  $T_i^\prime(a(|x|),x)$ is true (in the standard model).

\begin{definition}\label{def:alpha}
  Let $i, c \in \N$
  and let $\varphi = \varphi(x)$ be
  a $\hat\Sigma^{1,b}_1$-formula (with only one free
  variable~$x$, and in particular without free variables of the set
  sort). Define
\begin{equation*}
\alpha_\varphi^{i,c}\ :=\ \forall n{\in}\Log\ \exists a{\le} 2^{n^c}
  \forall x{<}2^n\ \big( T^\prime_i(a,x) \leftrightarrow \varphi(x)\big).
\end{equation*}
\end{definition}

It is clear that $\bigl\{\neg \alpha^{i,c}_\varphi\mid i, c\in \N \bigr\}$ is true if and only if the $\NEXP$-problem defined by $\varphi(x)$ does not belong to $\PHpoly$. Hence, the following states the consistency of $\NEXP\not\subseteq\PHpoly$:

\begin{theorem}\label{thm:deltaNEXP} There exists $\varphi(x)\in \hat\Sigma^{1,b}_1$
such that $\V^0_2+\bigl\{\neg \alpha^{i,c}_\varphi\mid i, c\in \N \bigr\}$
is consistent.
\end{theorem}

This
is proved in almost exactly the same way as
the just-given proof of Theorem~\ref{thm:alphaNEXP}.
The only difference is that, working in a model
of $\V^0_2 + \alpha^{i,c}_\varphi$,
the circuit~$C(x)$
is replaced with the formula $T^\prime_i(a,x)$ for
an advice string $a \le 2^{|x|^c}$.  The
details are left to the reader.

\section{Formally verified model-checkers} \label{sec:mc}

We shall need to formally reason about certain straightforwardly
defined exponential time machines, namely model-checkers and universal
machines. A model-checker~$M_\varphi$ for a
formula~$\varphi(\bar X,\bar x)$ has oracle access to~$\bar X$ and, on
input~$\bar x$, decides whether~$\varphi(\bar X,\bar x)$ is true. For
example, by nesting a loop for each bounded
quantifier,~$\Sigma^{1,b}_0$-formulas have straightforward
model-checkers that run in exponential time and polynomial space.  We
define such model-checkers with care, so that~$\S^1_2(\alpha)$
verifies their time and space bounds as well as their
correctness. This correctness statement has to be formulated carefully
because, in general,~$\S^1_2(\alpha)$ cannot prove that a halting
computation of~$M^{\bar X}_\varphi$ on~$\bar x$ exists. Thus, proving
correctness means to show that {\em if} a computation exists, {\em
  then} it does what it is supposed to do. To prove this we use some
constructions that are similar in spirit to those in~\cite{BeckmannBuss:localImprove}.

\subsection{Preliminaries: explicit machines}\label{sec:explicit}

In short, a machine will be called \emph{explicit} if the
theory~$\S^1_2(\alpha)$ proves that its halting computations terminate
within a specified number of steps, using no more than a specified
amount of space in its work tapes, and by querying its oracles no
further than a specified position.

\paragraph{Machine model.}
Our model of computation is the multi-tape oracle Turing machine with
one-sided infinite tapes (i.e., cells indexed by~$\N$) and an alphabet
containing~$\{0,1\}$. The content of cell~$0$ is fixed to a fixed
symbol marking the end of the tape. At the start, the heads scan
cell~$1$.  The machines can be deterministic or
non-deterministic. Such a machine~$M$ has read-only input tapes, and
work tapes and oracle tapes. If there are~$k$ input tapes, then its
inputs are~$k$-tuples~$\bar x=(x_1,\ldots, x_k)$ of numbers with the
length-$|x_i|$ binary representation of~$x_i$ written on the~$i$-th
input tape. The length of the input is~$|\bar
x|=\max_{i}|x_i|$. If~$M$ does not have oracle tapes, then it is a
machine {\it without oracles}. If~$M$ has~$\ell \geq 1$ oracle tapes,
then we write~$M^{\bar X}$ for the machine with
oracles~$\bar X=(X_1,\ldots,X_\ell)$.  When the machine enters a
special query state, it moves to one out of~$2^\ell$ many special
answer states which codes the answers to the~$\ell$ queries written on
the~$\ell$ oracle tapes, i.e., whether the number written (in binary)
on the~$i$-th oracle tape belongs to~$X_i$ or not.

A {\em partial space-$s$ time-$t$ query-$q$ computation
  of~$M^{\bar X}$ on~$\bar x$} comprises~$t+1$ configurations, the
first one being the starting configuration, every other being a
successor of the previous one, and repeating halting configurations,
if any.  Being {\em space-$s$} means that the largest visited cell on
each tape is at most~$s$, and being {\em query-$q$} means that the
largest visited cell on each oracle tape is at most most~$|q|$; in other words,
all queries have length at most~$|q|$. Query lengths are bounded by~$|q|$
instead of~$q$ so that all queries are restricted to have polynomial length.

\paragraph{Coding computations.} \label{sec:premachines} Fix a
machine~$M$. Let~$s,t,q \in \N$ and consider a partial space-$s$,
time-$t$, query-$q$ computation of~$M$ on an unspecified input with
unspecified oracles.  A configuration is coded by
an~$(s{+}1)$-tuple~$(q,c_0,\ldots,c_{s-1})$ of numbers:~$q$ codes the
current state of the machine;~$c_i$ codes,~for each tape, a {\em
  position bit} indicating whether the index of the currently scanned
cell is at most~$i$ and, for each work or oracle tape, the content of
cell~$i$. We assume that these numbers are smaller than~$M$ (the
machine is (coded by) a number), so we get an $(s{+}1)\times (t{+}1)$
matrix of such numbers.  This matrix is coded by the set~$Y$ of
numbers bounded by~$\langle s,t,|M|\rangle$ that contains exactly
those~$\langle i,j,k\rangle$ such that~$i \leq s$,~$j \leq t$,~$k<|M|$
and the~$(i,j)$-entry of the matrix has~$k$-bit~$1$.

The details of the encoding are irrelevant. What is required is that
there is a $\PV(\alpha)$-function~$f^Y$ such that $f^Y(t,s,q,j)$
gives, about the $j$-th configuration, a number coding the state, the
positions of the heads, the contents of the cells they scan, and the
numbers that are written in binary in the first~$|q|$ cells of the
oracle tapes. In the encoding sketched above, to find the position of
a specific head, $f^Y$~uses binary search to find~$i\le s$ where its
position bit flips; computing the oracle queries is possible because
the oracle tapes contain numbers below~$2^{|q|}$.
Having~$f^Y$, it is straightforward to write a
natural $\Pi^b_1(\alpha)$-formula stating
\begin{equation}\label{eq:formcomp}
  \q{$Y$ is a partial space-$s$ time-$t$ query-$q$ computation of~$M^{\bar X}$ on~$\bar x$}.
\end{equation}
The free variables of this formula are $Y,\bar X,\bar
x,s,t,q$. Exceptionally, we shall also consider~$M$ on the formal
level, in which case $M$~is an additional free {\em number variable}.
All quantifiers in the $\Pi^b_1(\alpha)$-formula \eqref{eq:formcomp}
can be $\S^1_2(\alpha)$-provably bounded by $p(s,t,|q|,|M|,|\bar x|)$
for a polynomial~$p$, where $|\bar x|$ stands for $|x_1|,\ldots,|x_k|$.
If~$M$ is a machine without oracles, the formula
is $\S^1_2(\alpha)$-provably equivalent to the one with~$q=0$, and we
omit `query-$q$'. We also omit `space~$s$' if~$s=t$. Further,
replacing `partial' by `halting' or `accepting' or `rejecting' are
obvious modifications of the formula.

\paragraph{Explicit machines.}\label{sec:explicitmachine}

Binary search gives a~$\PV(\alpha)$-function~$\textit{time}^Y(s,t)$
such that, provably in~$\S^1_2(\alpha)$,
if~$Y$ is a halting time-$t$
space-$s$ query-$q$ computation of~$M^{\bar X}$ on~$\bar x$,
then~$\textit{time}^Y(s,t)$ is the minimal~$j\le t$ such that
the~$j$-th configuration in~$Y$ is halting. We make the further
assumption that~$M$ never writes blank (but can write a copy of this symbol), so
heads leave marks on visited cells. Binary search can then
compute the maximal non-blank cell in the~$j$-th configuration on any
tape.  By quantifier-free induction
for~$\PV(\alpha)$-formulas, $\S^1_2(\alpha)$~proves that this cell number is
non-decreasing for~$j=0,1,\ldots, t$. Hence, there is
a~$\PV(\alpha)$-function~$\textit{space}^Y(s,t)$ such that, provably
in~$\S^1_2(\alpha)$, if~$Y$ is a halting time-$t$ space-$s$ query-$q$
computation of~$M^{\bar X}$ on~$\bar x$, then~$\textit{space}^Y(s,t)$
is the maximal cell visited in~$Y$ on any tape.  Similarly, there is
a~$\PV(\alpha)$-function~$\textit{query}^Y(s,t)$ that computes the
maximal cell visited on a query tape.

\begin{definition}
  A machine~$M$ is {\em explicit} if there are
  terms~$s(\bar x),t(\bar x),q(\bar x)$ such that
$$
\begin{array}{lcl}
\S^1_2(\alpha)&\vdash&
  \q{$Y$ is a halting
  space-$s'$ time-$t'$
  query-$q'$ computation of~$M^{\bar X}$ on~$\bar x$} \to \\
  &&\quad
  \textit{time}^Y(s',t')\le t(\bar x)\wedge
  \textit{space}^Y(s',t')\le s(\bar
  x)\wedge\textit{query}^Y(s',t')\le |q(\bar x)|.
\end{array}
$$
We say that the terms~$s=s(\bar x),t=t(\bar x),q=q(\bar x)$ {\em
  witness} that~$M$ is explicit. Further, if~$r(\bar x)$ is another
term, then we say that~$r=r(\bar x)$ {\em witnesses} that~$M$ is an

\medskip

\noindent
\begin{tabular}{llll}
{\em explicit $\NEXP$-machine} & if it is non-deterministic & with $t=s=q=r$; \\
{\em explicit $\EXP$-machine} & if it is deterministic & with $t=s=q=r$; \\
{\em explicit $\PSPACE$-machine} & if it is deterministic & with $t=q=r$
and $s=|r|$; \\
{\em explicit $\NP$-machine} & if it is non-deterministic & with $t=s=|r|$ and $q=r$; \\
{\em explicit $\P$-machine} & if it is deterministic & with $t=s=|r|$ and $q=r$.
\end{tabular}
\end{definition}

Observe that, if $s,t,q$ witness that $M$ is explicit,
and $s'=s'(\bar x)$, $t'=t'(\bar x)$, $q'=q'(\bar x)$ are terms such
that $\S^1_2 \vdash s(\bar x){\le}s'(\bar x) \wedge t(\bar
x){\le}t'(\bar x)\wedge q(\bar x){\le}q'(\bar x)$, then
also $s',t',q'$ witness that $M$ is explicit. E.g., if $r$ witnesses
that $M$ is an explicit~$\P$-machine, then $r$ also witnesses that $M$
is an explicit~$\PSPACE$-machine.

Given an explicit machine~$M$, we omit `space-$s$ time-$t$ query-$q$'
in \eqref{eq:formcomp} and its variations with `halting', `accepting'
or `rejecting'.  E.g. for an explicit~$\EXP$-machine~$M$, say
witnessed by~$r=r(\bar x)$, we have a~$\Pi^b_1(\alpha)$-formula
\begin{equation}
\q{$Y$ is an accepting computation of $M^{\bar X}$ on $\bar x$}.
\label{eqn:thepib1formula}
\end{equation}
This means that~$Y$ is a space-$r(\bar x)$ time-$r(\bar x)$
query-$r(\bar x)$ computation of~$M^{\bar X}$ on~$\bar x$ that ends in
an accepting halting configuration, and all queries \q{$z\in X$?}
during the computation satisfy~$z<2^{|r(\bar x)|}$. In particular,
\begin{equation}
Y{\le}\langle r(\bar x),r(\bar x),|M|\rangle \label{eqn:referredtolater}
\end{equation}
provably in~$\S^1_2(\alpha)$. Furthermore, all quantifiers in
the~$\Pi^b_1(\alpha)$-formula~\eqref{eqn:thepib1formula} can
be~$\S^1_2(\alpha)$-provably bounded by~$p(r(\bar x),|M|,|\bar x|)$
for a polynomial~$p$, where~$|\bar x|$ stands
for~$|x_1|,\ldots,|x_k|$.

Thereby, our mode of speech follows~\cite[Definition~8.1.2]{Krajicek:book}
in that the time bound is used to determine the bound on the oracle
tapes.

\paragraph{Polynomial-time computations.}
It is well-known that $\S^1_2$ formalizes polynomial time computations. We shall use this in the form of the following lemma.

For an explicit~$\P$-machine~$M$, its computations~$Y$ can be coded by
numbers~$y$ and we get a~$\Pi^b_1(\alpha)$-formula
$$
\q{$y$ is a halting computation of $M^{\bar X}$ on $\bar x$}.
$$
Here,~$y$ is a number sort variable, and the free variables
are~$\bar X,\bar x,y$. If~$M$ has a special output tape, we agree that
the output of a computation is the number whose binary representation
is written in cells~$1,2,\ldots$ up to the first cell not containing a
bit.  We have a~$\PV(\alpha)$-function~$\out_M$ such that, provably
in~$\S^1_2(\alpha)$, if~$y$ is a halting computation of~$M^{\bar X}$
on~$\bar x$, then~$\out_M(y,j)$ is the content of cell~$j$ of the
output tape in the halting configuration in case this is a bit;
otherwise~$\out_M(y,j){=}2$. In particular,~$\S^1_2(\alpha)$
proves~$\out_M(y,j){\leq}2$,

\begin{lemma}\label{lem:pvmachine}
  For every~$\PV(\alpha)$-function~$f^{\bar X}(\bar x)$ there are an
  explicit~$\P$-machine~$M$ and
  a~$\PV(\alpha)$-function~$g^{\bar X}(\bar x)$ such
  that~$\S^1_2(\alpha)$ proves
\begin{eqnarray*}
&&\big(\q{$y$ is a halting computation of $M^{\bar X}$ on $\bar x$} \leftrightarrow y{=}g^{\bar X}(\bar x)\big) \wedge\\
&&\big(j{<}|f^{\bar X}(\bar x)|\to \out_M(g^{\bar X}(\bar x),j{+}1){=}\bit(f^{\bar X}(\bar x),j)\big) \wedge\\
&&\big(j{\ge}|f^{\bar X}(\bar x)|\to \out_M(g^{\bar X}(\bar x),j{+}1){=}2\big).
\end{eqnarray*}
\end{lemma}
\noindent In the statement of the lemma,~$\bit(n,i)$ is
a~$\PV$-function computing the~$i$-bit of the binary representation
of~$n$, i.e.,~$\bit(n,i)=\floor{n/2^i} \ \mathrm{ mod }\ 2$ (in the
standard model). In particular, we have~$\bit(n,i)=0$ for~$i\ge |n|$.

\subsection{Deterministic model-checkers}

For every~$\Sigma^{1,b}_0$-formula~$\varphi = \varphi(\bar X,\bar{x})$
in the language~$\PV(\alpha)$ we define its \emph{bounding
term}~$\bt_{\varphi}(\bar x)$ as follows:

\begin{enumerate} \itemsep=0pt
\item $\bt_\varphi=0$ if~$\varphi$ is atomic,
\item $\bt_\varphi=\bt_\psi$ if $\varphi=\neg\psi$,
\item $\bt_\varphi=\bt_\psi+\bt_\theta$ if $\varphi=(\psi\wedge\theta)$,
\item $\bt_\varphi=\bt_\psi(\bar x,t(\bar x))+t(\bar x)$ if $\varphi=\exists y{\le}t(\bar x)\ \psi(\bar X,\bar x,y)$.
\end{enumerate}

\begin{lemma}\label{lem:mcdet}\
  For every $\Sigma^{1,b}_0$-formula $\varphi=\varphi(\bar X,\bar x)$
  there are an explicit $\PSPACE$-machine~$M_\varphi^{\bar X}$, a $\Sigma^{1,b}_0$-formula~$\Comp_\varphi(\bar X,\bar x,u)$,
  terms $r_\varphi(\bar x),s_\varphi(\bar x)$,
  and a
  polynomial $p_{\varphi}(m,\bar n)$, such that
\begin{enumerate}[label=(\alph*), ref=\alph*]\itemsep=0pt
\item $\S^1_2(\alpha)\vdash \q{$Y$ is an accepting computation of~$M_\varphi^{\bar X}$ on $\bar x$} \to \varphi(\bar X,\bar x)$,
\label{lem:mcdet.a}
\item $\S^1_2(\alpha)\vdash \q{$Y$ is a rejecting computation of~$M_\varphi^{\bar X}$ on $\bar x$} \to \neg\varphi(\bar X,\bar x)$,
\label{lem:mcdet.b}
\item $\S^1_2(\alpha)\vdash  \q{$\Comp_\varphi(\bar X,\bar x,\cdot)$ is a halting computation of~$M_\varphi^{\bar X}$ on $\bar x$}$,
\label{lem:mcdet.c}
\item $\S^1_2(\alpha)\vdash\ r_\varphi(\bar x) \leq
  p_{\varphi}(\bt_{\varphi}(\bar x),|\bar x|)$
\label{lem:mcdet.bound1},
\item $r_{\varphi}(\bar x),s_\varphi(\bar x)$
  witness~$M^{\bar X}_\varphi$ as explicit~$\EXP$- and $\PSPACE$-machines,
  respectively.
\label{lem:mcdet.bound2}
\end{enumerate}
In addition, if~$\varphi = \varphi(\bar X,\bar x)$ is
a~$\Pi^b_1(\alpha)$-formula, then there are a
term~$t_\varphi(\bar x)$ and a
quantifier-free~$\PV(\alpha)$-for\-mu\-la~$\Comp_\varphi(\bar{X},\bar{x},w,u)$
such that
\begin{enumerate}[resume*]\itemsep=0pt
\item $\T^1_2(\alpha)\vdash\
\exists w{\le}t_\varphi(\bar x)\
\q{$\Comp_\varphi(\bar X,\bar x,w, \cdot)$ is a halting
computation of~$M_\varphi^{\bar X}$ on $\bar x$}$,
\label{lem:mcdet.f}
\item $\S^1_2(\alpha)\vdash\ \varphi(\bar X,\bar x)\to $
\hbox to 2in{\q{$\Comp_\varphi(\bar X,\bar x,t_\varphi(\bar x), \cdot)$ is
an accepting computation of~$M_\varphi^{\bar X}$ on $\bar x$}}.
\label{lem:mcdet.g}
\end{enumerate}
\end{lemma}

\begin{proof}
  Call a $\Sigma^{1,b}_0$-formula~$\varphi = \varphi(\bar X,\bar x)$
  \emph{good} if it satisfies~(a)--(e). Observe that
  all $\Sigma^b_0(\alpha)$-formulas are good: they
  are $\S^1_2(\alpha)$-provably equivalent to formulas of the
  form $f^{\bar X}(\bar x){=}1$ for
  some $\PV(\alpha)$-function~$f^{\bar X}(\bar x)$, and we can choose
  a machine according to Lemma~\ref{lem:pvmachine}. Recall that an
  explicit $\P$-machine is also an explicit $\PSPACE$-machine and
  explicit $\EXP$-machine (in this case, all three witnessed by the same term).

We leave it to the reader to check that the good formulas are closed
under Boolean combinations.  We are then left to show that if
\begin{equation}\label{eq:phi}
\varphi(\bar X,\bar x)\ =\ \exists y{\le}t(\bar
x)\ \psi(\bar X,\bar x,y)
\end{equation}
for a term~$t(\bar x)$ and a good formula~$\psi = \psi(\bar X,\bar
x,y)$, then~$\varphi$ is good.  To lighten the notation, in the
following we drop any reference to the set-parameters~$\bar X$ in the
formulas, and to the oracles~$\bar X$ in machines, since they remain
fixed throughout the proof.

The machine~$M_\varphi$ runs a loop searching for a~$y$
in~$\{0,\ldots,t(\bar x)\}$ that satisfies~$\psi$.
On input~$\bar x$, it writes~$y:=0$ on a work tape and then loops: it
checks whether~$y\le t(\bar x)$ and, if so, it updates~$y:=y+1$ and
runs~$M_\psi$ on~$(\bar x, y)$; otherwise it halts. It accepts or
rejects according to a {\em flag} bit~$b$ stored in its state
space:~$b$ is initially set to~$0$, and it is set to 1 when and if
an~$M_\psi$-run accepts.

To prove~(a)--(e) we want a
quantifier-free~$\PV(\alpha)$-formula~$D(Y,\bar x,y,u)$ that extracts
the~$M_\psi$-computation simulated in the~$y$-loop. More precisely,
we want~$\S^1_2(\alpha)$ to prove that, if~$Y$ is a halting
computation of~$M_\varphi$ on~$\bar x$, then~$D(Y,\bar x,y,\cdot)$ is
a halting computation of~$M_\psi$ on~$(\bar x,y)$.  For this, we
design the details of~$M_\varphi$ in a way so that the~$j$-th step of
the computation of~$M_\psi$ on~$(\bar x,y)$ is simulated
by~$M_\varphi$ at a time easily computed from~$\bar x,y,j$.

\medskip

{\em Description of~$M_\varphi$.}
Set~$r(\bar x):=r_\psi(\bar x,t(\bar x))$ where~$r_\psi(\bar x,y)$ is
the term claimed to exist for~$\psi$. Note that~$\S^1_2(\alpha)$
proves that~$r_\psi(\bar x,y)\le r(\bar x)$ for~$y \le t(\bar
x)$. Additionally to properties~(a)--(e) for~$\psi$, we assume
inductively that~$\S^1_2(\alpha)$ proves that the halting
configuration of~$M_\psi$ on~$(\bar x,y)$ equals the initial
configuration except for the state, that is,~$M_\psi$ cleans all
worktapes and moves all heads back to cell~$1$ before it halts.

Our machine initially computes~$t = t(\bar x)$ and~$r=r(\bar x)$ and
two binary {\em clocks} initially set to~$0^{|t|}$ and~$0^{|r|}$. The
terms are evaluated using explicit~$\P$-machines according to
Lemma~\ref{lem:pvmachine}.  The initial settings of the clocks are
simply computed by scanning the binary representations of~$t$ and~$r$
that were computed at the start. This initial computation of terms,
and initialization of clocks, takes time exactly~$\textit{ini}(\bar
x)$ for some~$\PV$-function~$\textit{ini}(\bar
x)$. Further,~$\S^1_2(\alpha)$ proves~$\textit{ini}(\bar x) \leq
|t_\textit{i}(\bar x)|$ for a suitable term~$t_\textit{i}(\bar x)$.

The~$y$-loop is implemented as follows. First update~$y$, the value of
the first clock. To do this, sweep over the first clock, and then
back, in exactly~$(2|t|+2)$ steps, doing the following: copy~$y$
without leading~$0$'s to some tape, so this tape holds the
length-$|y|$ binary representation of~$y$ (as expected by~$M_\psi$);
increase the clock by~$1$ if~$y<t$, and reset it to~$0^{|t|}$
if~$y=t$; in the latter case store a bit signaling this; this signal
bit halts the computation (in the next~$y$-loop) instead of doing
the~$y$-update.
After this~$y$-update, simulate~$r$ steps of~$M_\psi$ on~$(\bar x,y)$
by an inner loop: in~$2|r|+2$ steps sweep twice over the second
clock. If its value was smaller that~$r$, then increase it by~$1$ and
simulate the next step of~$M_\psi$'s computation; this can mean
repeating the halting computation.  If its value was not smaller
than~$r$, then set the clock back to~$0^{|r|}$.
Thus, exactly~$2|r|+3$ steps are spent for one step of~$M_\psi$ and
one~$y$-loop takes exactly~$t_{\ell}(\bar x) := (r(\bar x)+1)
\cdot (2|r(\bar x)|+3)$ steps.

If the signal bit halts the computation, then our machine first cleans
all tapes and moves heads back to cell~$1$, before halting. We omit a
description of this final polynomial time computation. It can be
implemented to take exactly~$\mathit{fin}(\bar{x})$ steps for
a~$\PV$-function~$\mathit{fin}(\bar x)$, and~$\S^1_2$
proves~$\mathit{fin}(\bar{x}) \leq |t_{\textit{f}}(\bar x)|$ for a
suitable term~$t_{\textit{f}}(\bar x)$.

Thus~$M_\varphi$ runs in time
exactly~$\textit{ini}(\bar x)+(t(\bar x)+1)\cdot t_{\ell}(\bar
x)+\textit{fin}(\bar x)$.  It simulates~$r$ steps of~$M_\psi$
on~$(\bar x,y)$ at times
\begin{equation}
t(\bar x,y,j):=\mathit{ini}(\bar x)+y\cdot t_{\ell}(\bar
x)+(j+1)\cdot (2|r(\bar x)|+3) \label{eqn:times}
\end{equation}
for~$j<r(\bar x)$.

\medskip

\textit{Explicitness: proof of (d)--(e).}
 Let~$s_{\psi}(\bar x,y)$ be the term that
witnesses~$M_\psi$ as an explicit~$\PSPACE$-machine.  Let~$Y$ be a
halting computation of~$M_\varphi$ on~$\bar x$. There is
a~$\PV(\alpha)$-function that from~$\bar x$ computes (a number coding)
the initial computation of terms and clocks, and~$\S^1_2(\alpha)$
proves its halting configuration is as
described. Clearly,~$\S^1_2(\alpha)$ proves that the
first~$\textit{ini}(\bar x)$ steps of~$Y$ coincide with this
computation. In particular,~$\S^1_2(\alpha)$ proves that the clocks
computed in~$Y$ have the desired length. Similarly, there is
a~$\PV(\alpha)$-function that from~$\bar x,y,j$ computes (a number
coding) the space-$|s_{\psi}(\bar x,y)|$ configuration of~$M_\psi$ at
time~$t(\bar x,y,j)$ in~$Y$.

We prove, by quantifier-free induction, that the computation~$Y$
simulates the steps of~$M_\psi$ at times~$t(y,j) := t(\bar x,y,j)$
for~$y\leq t$ and~$j<r$.  Assume this holds for time~$t(y,j)$. We
verify it for time~$t(y,j+1)$ or time~$t(y+1,0)$ depending on
whether~$j<r$ or~$j=r$. Assume the former; the latter case is
similar. Compute the time-$(2|r|+3)$ computation (that sweeps twice
over the clock and simulates one more step of~$M_{\psi}$) starting at
the configuration at time~$t(y,j)$; then~$Y$ must coincide with this
computation between time~$t(y,j)$ and time~$t(y,j+1)$.  Hence,~$Y$
simulates a step of~$M_\psi$ at time~$t(y,j+1)$. Similarly,
quantifier-free induction proves that the~$M_\psi$-configurations at
the times~$t(y,j)$ in~$Y$ are successors of each others.  This yields a
quantifier-free~$\PV(\alpha)$-formula~$D(Y,\bar{x},y,u)$ as desired.

From the configuration at time~$\textit{ini}(\bar x)+(t+1)\cdot
t_{\ell}(\bar x)$ one can compute the final~$\textit{fin}(\bar
x)$ steps of the clean-up computation before~$M_\varphi$ halts, and
the last~$\textit{fin}(\bar x)$ steps of~$Y$ must coincide with
that. Hence,~$\S^1_2(\alpha)$ proves that the configuration of~$Y$ at
time~$\textit{ini}(\bar x)+(t+1)\cdot t_{\ell}+\textit{fin}(\bar
x)$ is halting. 	
Recalling that~$\textit{ini}(\bar x) \leq
|t_{\textit{i}}(\bar x)|$ and~$\textit{fin}(\bar x) \leq
|t_{\textit{f}}(\bar x)|$, this implies that the term
\begin{eqnarray*}
r_\varphi(\bar x) & := &
|t_{\textit{i}}(\bar x)|+(t(\bar x)+1)\cdot
t_{\ell}(\bar x)+ |t_{\textit{f}}(\bar x)|
\end{eqnarray*}
witnesses~$M_{\varphi}$ as an explicit $\NEXP$-machine. Choose a term $s_\varphi(\bar x)$ such that $\S^1_2$-provably $s_\varphi(\bar x)\ge  r_\varphi(\bar x)$ and
$$
|s_\varphi(\bar x)|\ge |t_{\textit{i}}(\bar x)|+(|t(\bar x)|+1)+(|r(\bar x)|+1)+
|s_{\psi}(\bar x,t(\bar x))|+|t_{\textit{f}}(\bar x)|.
$$
Then $s_\varphi(\bar x)$ witnesses~$M_{\varphi}$ as an explicit $\PSPACE$-machine. This shows (e).

For~(d), recall $t_{\ell}(\bar x) = (r(\bar x)+1)\cdot (2|r(\bar x)|+3)$ and
hence $r_{\varphi}(\bar x) \le p(r(\bar x),t(\bar x),|\bar x|)$ for a
suitable polynomial~$p$, provably in~$\S^1_2$. Recalling that $r(\bar x) = r_{\psi}(\bar x,t(\bar x))$, and that by~(d)
for~$\psi$ we
have $r_{\psi}(\bar x,y) \leq p_{\psi}(\bt_{\psi}(\bar x,y),|\bar
x|,|y|)$ provably in~$\S^1_2$,
from $\bt_{\varphi}(\bar x) = \bt_{\psi}(\bar x,t(\bar x)) + t(\bar
x)$ we get, also provably in~$\S^1_2$,
that $r_{\varphi}(\bar x) \leq p_{\varphi}(\bt_{\varphi}(\bar x),|\bar
x|)$ for a suitable polynomial~$p_\varphi$.
%
%

\medskip

\textit{Correctness: proof of (a)--(c).}
For~(a) argue in~$\S^1_2(\alpha)$ and
suppose~$Y$ is an accepting computation of $M_\varphi$ on~$\bar
x$. Being accepting means that the final state has flag $b=1$, while
the starting state has flag $b=0$.  By binary search we find a time
when~$b$ flips from $0$ to~$1$. This time determines $y_0\le t$ such
that the $y_0$~loop accepts. Then $Z:=D(Y,\bar{x},y_0,\cdot)$ is an
accepting computation of~$M_\psi$ on~$(\bar x, y_0)$.  Note that $Z$
exists by $\Delta^b_1(\alpha)$-comprehension. Then~(a) for~$\psi$
implies $\psi(\bar x,y_0)$ and thus~$\varphi(\bar x)$.

For~(b), argue in~$\S^1_2(\alpha)$ and suppose~$Y$ is a rejecting
computation of~$M_\varphi$ on~$\bar x$, so the flag is~$0$ in the
final configuration.  Let~$y\le t$. Then~$D(Y,\bar{x},y,\cdot)$ is a
rejecting computation of~$M_\psi$ on~$(\bar x, y)$: otherwise the~$y$
loop sets the flag to~$1$ and then binary search finds a time where
the flag flips from~$1$ to~$0$ in~$Y$ which contradicts the working
of~$M_\varphi$.  Then~(b) for~$\psi$ implies~$\neg\psi(\bar
x,y)$. As~$y$ was arbitrary, we get~$\neg\varphi(\bar x)$.

For~(c), it is easy to construct from~$\Comp_\psi$ a
formula~$\Comp_{\psi,0}$ such that~$\S^1_2(\alpha)$ proves that the
set~$\Comp_{\psi,0}(\bar x,y,\cdot)$ is the computation of
the~$y$-loop of~$M_\varphi$ on~$\bar x$ with flag~$0$ stored in the
state space. There is an analogous formula~$\Comp_{\psi,1}$ for
flag~$1$. These formulas just stretch the computation described
by~$\Comp_\psi$ and interleave it with the trivial updates of the
clocks.  The desired formula~$\Comp_\varphi(\bar x,u)$ `glues
together' these computations, plus the initial~$\textit{ini}(\bar x)$
steps of initialization, and the final~$\textit{fin}(\bar x)$ steps of
clean-up. We sketch the definition of~$\Comp_\varphi(\bar x,u)$:
from~$u$ we can compute~$y$ such that the truth value
of~$\Comp_\varphi(\bar x,u)$ is one of the bits in the code of the
computation of the~$y$-loop of~$M_\varphi$ on~$\bar x$, or one of the
bits in the code of the initial or final
computation. Then~$\Comp_\varphi(\bar x,u)$ states
\begin{equation}\label{eq:hat}
\begin{array}{lcl}
( \exists z{<}y\ \psi(\bar x,z)\wedge\Comp_{\psi,1}(\bar x,y,u)) \vee
  ( \neg\exists z{<}y\ \psi(\bar x,z)\wedge\Comp_{\psi,0}(\bar
  x,y,u)).
\end{array}
\end{equation}
%
%

\medskip

\textit{Proof of (f)--(g).} Assume~$\varphi$ is
a~$\Pi^b_1(\alpha)$-formula. We modify the given construction as
follows. Up to~$\S^1_2(\alpha)$-provable equivalence we have
\begin{equation*}
\varphi(\bar{X},\bar{x}) = \forall y{\le}t(\bar x)\ g^{\bar X}(\bar
x,y){=}1
\end{equation*}
where~$t(\bar x)$ is a term and~$g^{\bar X}(\bar x,y)$ is
a~$\PV(\alpha)$-function. As before, we drop any reference to the
set-parameters~$\bar X$, and to the oracles~$\bar X$, since they will
stay fixed throughout the proof. We define~$M_{\varphi}$ similarly as
before with the role of~$M_\psi$ played by a~$\P$-machine
checking~$g(\bar x,y){=}1$ according to Lemma~\ref{lem:pvmachine}.
The only difference is in the flag bit: it is initially set to~$1$,
and it is set to~$0$ when and if a~$y$-loop rejects
(meaning~$\neg g(\bar x,y){=}1$).

In this case we can choose~$r$ small, i.e., equal to~$|r'|$ for some
term~$r'=r'(\bar x)$, so there is
a $\PV(\alpha)$-function~$h(\bar x,y)$ that computes (a number that
codes) the computation of the $y$-loop of~$M_\varphi$.
Then $\Comp_\varphi(\bar x,w,u)$ `glues together' these computations
plus suitable initial and final computations.  The only problem is to
determine the flag~$b$ stored in the states of~$M_\varphi$. For this
we need to know the minimal~$w\le t$ such that $\neg g(\bar x,w){=}1$
holds, or take $w=t+1$ if $\varphi(\bar x)$ holds.  Such~$w$ exists
provably in~$\T^1_2(\alpha)$. This shows~(f)
for $t_\varphi(\bar x):=t(\bar x)+1$. For~(g),
assuming~$\varphi(\bar x)$ we can take $w=t+1$ directly since in this
case the flag bit is always~1 provably in~$\S^1_2(\alpha)$.
\end{proof}

\begin{remark}
The proof shows that the quantifier complexity of~$\Comp_\varphi$ is
close to that of~$\varphi$. If~$\varphi\in\Sigma^b_{0}(\alpha)$,
then~$\Comp_\varphi$ is a quantifier
free~$\PV(\alpha)$-formula. If~$\varphi\in\Sigma^b_{i}(\alpha)$
for~$i>0$, then~$\Comp_\varphi$ is a Boolean combination
of~$\Sigma^b_{i}(\alpha)$-formulas. Note that if the outer quantifier
in~\eqref{eq:phi} is sharply bounded, i.e.,~$t(\bar x)=|t'(\bar x)|$
for some term~$t'(\bar x)$, then the~$y$-bounded quantifiers
in~\eqref{eq:hat} are sharply bounded too.
\end{remark}

\subsection{Optimality remarks}

This subsection offers some remarks stating that
Lemma~\ref{lem:mcdet}.\ref{lem:mcdet.f} cannot be improved in certain
respects. This material is not needed in the following.

\begin{remark}
For our definition of~$M^{\bar X}_\varphi$, one cannot
replace~$\T^1_2(\alpha)$ by~$\S^1_2(\alpha)$ in
Lemma~\ref{lem:mcdet}.\ref{lem:mcdet.f} unless~$\S^1_2=\T^1_2$.
\end{remark}

\begin{proof}
Let~$\varphi(x)=\exists y{\le} x\ \psi(y,x)$ for~$\psi$ a
quantifier-free~$\PV$-formula, and assume \eqref{lem:mcdet.f} holds
for~$\S^1_2(\alpha)$ instead of~$\T^1_2(\alpha)$. We
show~$\S^1_2(\alpha)$ proves that, if there is~$y\le x$ such
that~$\psi(y,x)$, then there is a minimal such~$y$.  Argue
in~$\S^1_2(\alpha)$ and
suppose~$\varphi(x)$. By~$\Delta^b_1(\alpha)$-comprehension and
\eqref{lem:mcdet.f} there is a halting computation~$Y$ of~$M_\varphi$
on~$x$. By \eqref{lem:mcdet.b} it cannot be rejecting, so is
accepting.  Our proof of \eqref{lem:mcdet.a} gives~$\psi(y_0,x)$
for~$y_0\le x$ such that the flag~$b$ flips from 0 to 1 in loop~$y_0$.
We claim~$y_0$ is minimal. This is clear if~$y_0=0$. Otherwise we
had~$b=0$ after the loop on~$y_0-1$ (in~$Y$). For contradiction,
assume there is~$y_1<y_0$ with~$\psi(y_1,x)$. Then the loop on~$y_1$
would set~$b=1$. By quantifier-free induction we find a time
between~$y_1$ and~$y_0-1$ where~$b$ flips from 1 to 0. This
contradicts the working of~$M_\varphi$.
\end{proof}

Fix {\em any} machines~$M_\varphi$ satisfying the lemma. Call a
formula {\em true} if its universal closure is true in the standard
model.

\begin{remark}
In Lemma~\ref{lem:mcdet}.\ref{lem:mcdet.f} the auxiliary~$\exists w$
cannot be omitted. There is
a~$\Sigma^{b}_1(\alpha)$-formula~$\varphi(X, x)$ such that for all
quantifier-free~$\PV(\alpha)$-formulas~$C(X,x,u)$ the following is not
true:
\begin{equation*}
\q{$C(X,x,\cdot)$ is a halting computation of~$M_\varphi^{X}$ on $x$}.
\end{equation*}
\end{remark}

\begin{proof}
Otherwise every $\Sigma^b_1(\alpha)$-formula~$\varphi(X,x)$ is
equivalent to a quantifier-free $\PV(\alpha)$-formula~$D(X,x)$.
Let $A\subseteq\N$ be such that $\NP^A\not\subseteq\P^A$ and
choose $Q$ in $\NP^A\setminus\penalty10000\P^A$. Choose
a $\Sigma^b_1(\alpha)$-formula~$\varphi(X,x)$ defining $Q$
in~$(\N,A)$, the model where~$X$ is interpreted by~$A$. Note $D(X,x)$
defines in~$(\N,A)$ a problem
in~$\P^A$. Then $(\varphi(X,x)\leftrightarrow\penalty10000 D(X,x))$ fails
in $(\N,A)$ for some~$x$, and hence also in $(\N,A')$ for some
bounded $A'\subseteq A$ (Remark~\ref{rem:setbound}). Thus, this
equivalence is not true.
\end{proof}

\begin{remark}
Lemma~\ref{lem:mcdet}.\ref{lem:mcdet.f} does not extend to much more
complex formulas. There is a~$\Pi^{b}_2(\alpha)$-formula~$\varphi(X,
x)$ such that for all terms~$t$ and all
quantifier-free~$\PV(\alpha)$-formulas~$C$ the following is not true:
\begin{equation*}
\exists w{\le}t(x)\q{$C(X,x,w, \cdot)$ is a halting computation
  of~$M_\varphi^{X}$ on $x$}.
\end{equation*}
\end{remark}

\begin{proof}
Note this is a~$\Sigma^b_2(\alpha)$-formula, so for
every~$A\subseteq\N$ defines in~$(\N,A)$ a problem
in~$(\Sigma^\P_2)^A$.  Choose~$A$ such that~$(\Pi^\P_2)^A\neq
(\Sigma^\P_2)^A$ and argue similarly as before.
\end{proof}

\subsection{Non-deterministic model-checkers}

We shall also need model-checkers
for~$\hat\Sigma^{1,b}_1$-formulas. As a first step we prove a
technical lemma showing how to convert an explicit
oracle~$\PSPACE$-machine~$M^Y$ into an explicit~$\NEXP$-machine~$N$
that first guesses the oracle~$Y$ on a \emph{guess tape}, and then
simulates~$M^Y$. As usual, we need to show that~$\S^1_2(\alpha)$ is
able to prove that this construction does what is claimed.

\begin{lemma} \label{lem:mcnondet}\ For every
  explicit~$\PSPACE$-machine~$M^{Y,\bar X}$ that, as
  explicit~$\EXP$-machine, is witnessed by term~$r_M(\bar x)$, there
  are an explicit~$\NEXP$-machine~$N^{\bar X}$, a term~$r_N(\bar x)$,
  a polynomial~$p_N(m,\bar n)$, and
  quantifier-free~$\PV(\alpha)$-formulas~$F,G,H$ such that
\begin{enumerate}[label=(\alph*), ref=\alph*]\itemsep=0pt
\item $\begin{array}[t]{lcl}
\S^1_2(\alpha)&\vdash&
\q{$Z$ is an accepting  computation of~$M^{Y,\bar X}$ on $\bar x$} \to \\
&&\q{$F(Z,Y,\bar X,\bar x,\cdot)$ is an accepting
computation of~$N^{\bar X}$ on $\bar x$}.
\end{array}$  \label{lem:mcnondet.a}
\item $\begin{array}[t]{lcl}
\S^1_2(\alpha)&\vdash&
\q{$Z$ is an accepting  computation of~$N^{\bar X}$ on $\bar x$} \to \\
&&\q{$G(Z, \bar X,\bar x,\cdot)$ is an accepting
computation of~$M^{H(Z,\bar X,\bar x,\cdot), \bar X}$ on $\bar x$}
\end{array}
$  \label{lem:mcnondet.b}
\item $\S^1_2(\alpha)\vdash\ r_N(\bar x)\le p_N(r_M(\bar x),|\bar x|)$,
\label{lem:mcnondet.c}
\item The term $r_N(\bar x)$ witnesses $N^{\bar X}$ as explicit $\NEXP$-machine.
\label{lem:mcnondet.d}
\end{enumerate}
\end{lemma}

\begin{proof}
  Set~$r=r_M(\bar x)$. By assumption, the triple of
  terms~$r_M(\bar x),r_M(\bar x),r_M(\bar x)$ witnesses
  that~$M^{Y,\bar X}$ is explicit. In particular, every
  query~$\q{$z\in Y$?}$ made by~$M^{Y,\bar X}$ on~$\bar x$
  satisfies~$|z| \leq |r|$ and hence~$z <2^{|r|}$. The
  machine~$N^{\bar X}$ on~$\bar x$ guesses a binary string~$Y$ of
  length~$2^{|r|}$ on a {\em guess tape} and then
  simulates~$M^{Y,\bar X}$ on~$\bar x$ as follows: an oracle
  query~$\q{$z\in Y$?}$ of~$M^{Y,\bar X}$ is answered reading
  cell~$z{+}1$ on the guess tape. As in the proof of
  Lemma~\ref{lem:mcdet}, to prove~(a)--(d) we need to design the
  details of~$N$ in a way so that the~$j$-th step of the computation
  of~$M$ is simulated by~$N$ at a time easily computed
  from~$\bar x,j$. To reduce notation, in the following we drop any
  reference to the oracles~$\bar X$ as they will remain fixed
  throughout the proof.

\medskip

{\em Description of~$N$.}  The machine~$N$ on~$\bar x$ first
computes~$r$ and two binary clocks initialized to~$0^{|r|+1}$
and~$0^{|r|}$, respectively.  To write~$Y$ of length~$2^{|r|}$ on the
guess tape the machine checks whether the first clock equals~$2^{|r|}$
and, if not, increases it by one and moves one cell to the right on
the guess tape. This is done in exactly~$2|r|+5$ steps.  Once the
clock equals~$2^{|r|}$, the machine moves back to cell~$1$ on the
guess tape and non-deterministically writes~$0$ or~$1$ in each step,
except in the step that finally rebounds on cell~$0$ to cell~$1$. The
terms are computed with explicit~$\P$-machines according to
Lemma~\ref{lem:pvmachine}.  The initial computation of terms, and
initialization of clocks, takes time exactly~$\mathit{ini}(\bar x)$
for some~$\PV$-function~$\mathit{ini}(\bar x)$.  Therefore, the guess
of~$Y$ takes
exactly~$\mathit{guess}(\bar x) := \mathit{ini}(\bar x) +
2^{|r|}\cdot(2|r|+5)+2^{|r|}+1$ steps. Moreover,~$\S^1_2$
proves~$\mathit{guess}(\bar x) \leq t_{\textit{g}}(\bar x)$,
where
\begin{equation*}
t_{\textit{g}}(\bar x) :=
|t_{\textit{i}}(\bar x)|+ 2^{|r_M(\bar x)|}\cdot(2|r_M(\bar
x)|+5)+2^{|r_M(\bar x)|}+1,
\end{equation*}
for a suitable term~$t_{\textit{i}}(\bar x)$ such that~$\S^1_2$
proves~$\mathit{ini}(\bar x) \leq |t_{\textit{i}}(\bar x)|$.

The machine simulates~$r$ steps of~$M^Y$ using the second
clock. Comparing this clock with~$r$ and updating it takes~$2|r|+2$
steps.  If the value of the clock is less than~$r$, then a step
of~$M^Y$ is simulated by reading the~$(z{+}1)$-cell of the guess tape
where~$z$ is the content of~$M^Y$'s oracle tape for~$Y$. This is done
as follows. The machine moves forward over the guess tape, and rewinds
back to cell~$1$. With each step forward it increases the first clock
by one and checks whether it equals~$z$ or~$2^{|r|}$. If and when the
clock equals~$z$, it stores the {\em oracle bit} read on the guess
tape in its state space. Otherwise, i.e.,~$z{\geq}2^{|r|}$, the
machine stores oracle bit~$0$. When the clock equals~$2^{|r|}$, the
scan of the guess tape ends, and the rewinding to cell~$1$ starts (in
the next step). Doing this takes time
exactly~$2^{|r|}\cdot(2|r|+4)+2^{|r|}+1$ and the oracle bit is stored
at time~$\min\{z,2^{|r|}\}\cdot (2|r|+4)$.  Thus, when the value of
the second clock is less than~$r$, one step of~$M^Y$ is simulated in
exactly
\begin{equation*}
t_{\textit{s}}(\bar x) := (2|r_M(\bar x)|+2)+2^{|r_M(\bar x)|}
\cdot(2|r_M(\bar x)|+4)+2^{|r_M(\bar x)|}+2
\end{equation*}
steps. Otherwise, the simulation halts in an accepting or rejecting
state according to~$M^Y$'s state. In total, the machine runs for
exactly~$\mathit{guess}(\bar x)+r\cdot t_{\textit{s}}(\bar
x)+(2|r|+2)$ steps. The steps of~$M^Y$
on~$\bar x$ are simulated at times
\begin{equation*}
t(\bar x,j) := \mathit{guess}(\bar x)+(j+1)\cdot t_{\textit{s}}(\bar x)
\end{equation*}
for~$j<r_M(x)$. The runtime is bounded by the
term
\begin{equation*}
r_N(\bar x) := t_{\textit{g}}(\bar x)+r_M(\bar x)\cdot t_{\textit{s}}(\bar x)
+ (2|r_M(\bar x)|+2)
\end{equation*}

\medskip

{\em Explicitness.}  We argue that this bound on the runtime of~$N$ can
be verified in~$\S^1_2(\alpha)$, given a halting computation~$Z$
of~$N$ on~$\bar x$. Note that, unlike the simulation in
Lemma~\ref{lem:mcdet}, a single step is simulated in possibly
exponential time~$t_{\textit{s}}(\bar x)$. However, this possibly
exponential time computation is simply described: Since $M^Y$ is an
explicit $\PSPACE$-machine, its configurations can be coded by
numbers. Now, given a number coding the configuration of~$M^Y$
within~$Z$ at time $t(j) := t(\bar x,j)$, say with $Y$-oracle
query~$z$, and given a time $i<t_{\textit{s}}(\bar x)$, we can compute
the configuration of the clocks and the state of the (to-be-)stored
oracle-bit at time $t(j) + i$. Now, quantifier-free induction suffices
to prove that the oracle bit is stored at the desired time and equals
the content of the $(z{+}1)$-cell of the guess
tape (or 0 if $z \geq 2^{|r|}$). Quantifier-free induction proves that
the configurations of~$M^Y$ within~$Z$ at times~$t(j)$ for $j<r$ are
successors of those preceding them.  In particular, $\S^1_2(\alpha)$
proves that the configuration at time~$r_N(\bar x)$ is halting. Space
and query bounds can be similarly verified, so $N$ is explicit and
witnessed by~$r_N(\bar x)$.

\medskip

{\em Proof of (a)--(d).}  For~(a), the quantifier-free formula~$F$
concatenates an initial polynomial-time computation of the terms and
clocks, a guess of~$Y$, and a simulation of~$Z$. Each configuration of
the guess of~$Y$ is computable in polynomial time. The simulation
of~$Z$ stretches each step of~$M^Y$ to a time~$t_\textit{s}(\bar x)$
computation, each configuration of which is easily computed from~$Y$
and~$Z$ in polynomial time. Quantifier-free induction proves that
a~$Y$-query~$z$ in~$Z$ is answered according to the bit in
the~$(z{+}1)$-cell on the guess tape.

For~(b), the quantifier-free formula~$H$ extracts the guess~$Y$
from~$Z$ and the quantifier-free formula~$G$ extracts the simulated
computation at the times~$t(\bar x,j)$ for~$j<r_M(\bar x)$.

For~(c) and~(d), we already argued that the term~$r_N(\bar x)$
witnesses~$N$ as an explicit~$\NEXP$-machine. The claim
that~$r_N(\bar x) \leq p_N(r_M(\bar x),|\bar x|)$ holds for a suitable
polynomial~$p_N$ follows by inspection, and~$\S^1_2(\alpha)$ proves
it.
\end{proof}

Now we can state the lemma that proves that
every $\hat\Sigma^{1,b}_1$-formula has a formally verified
model-checker. In its statement, the bounding term~$\bt_\psi(\bar x)$
of a $\hat\Sigma^{1,b}_1$-formula $\psi = \psi(\bar X,\bar x)$ as in
Equation~\eqref{eqn:generichatSigma} is defined
to be the bounding term~$\bt_{\varphi}(\bar x)$ of its
maximal $\Sigma^{1,b}_0$
subformula $\varphi = \varphi(Y,\bar X,\bar x)$.

\begin{lemma}\label{lem:mc} For
  every~$\hat\Sigma^{1,b}_1$-formula~$\psi = \psi(\bar X,\bar x)$,
  there exists an explicit~$\NEXP$-machine~$N^{\bar X}_\psi$, a
  term~$r_\psi(\bar x)$, and a polynomial~$p_\psi(m,\bar n)$, such
  that
\begin{enumerate}[label=(\alph*), ref=\alph*]\itemsep=0pt
\item
  $\V^0_2 \vdash\ \psi(\bar X,\bar x) \to \exists_2 Y
  \q{$Y$ is an accepting computation of $N^{\bar X}_\psi$ on $\bar
    x$}.$  \label{lem:mc.a}
\item
  $\S^1_2(\alpha)\vdash\ \neg\psi(\bar X,\bar x) \to \neg\exists_2 Y
  \q{$Y$ is an accepting computation of $N^{\bar X}_\psi$ on $\bar
    x$}.$
\label{lem:mc.b}
\item $\S^1_2(\alpha)\vdash\ r_\psi(\bar x) \leq
  p_{\psi}(\bt_{\psi}(\bar x),|\bar x|)$,
\label{lem:mc.c}
\item the term~$r_\psi(\bar x)$
  witnesses~$N^{\bar X}_\psi$ as
  explicit~$\NEXP$-machine.
\label{lem:mc.d}
\end{enumerate}
Furthermore, if the maximal~$\Sigma^{1,b}_0$-subformula of~$\psi$ is
a~$\Pi^b_1(\alpha)$-formula, then
\begin{enumerate}[resume*]\itemsep=0pt
\item
  $\S^1_2(\alpha) \vdash \psi(\bar X,\bar x) \leftrightarrow \exists_2
  Y \q{$Y$ is an accepting computation of $N^{\bar X}_\psi$ on $\bar
    x$}.$  \label{lem:mc.e}
\end{enumerate}
\end{lemma}

\begin{proof}
  Let~$\psi(\bar X,\bar x) = \exists_2 Y\ \varphi(Y,\bar X,\bar x)$
  where~$\varphi = \varphi(Y,\bar X,\bar x)$ is
  a~$\Sigma^{1,b}_0$-formula. Recall that the bounding term of~$\psi$
  is~$\bt_\psi(\bar x) = \bt_\varphi(\bar x)$. In what follows, to
  lighten the notation, we drop any reference to the set
  parameters~$\bar X$ in formulas, and to the oracles~$\bar X$ in
  machines, since they remain fixed throughout the proof.

Let~$M^Y_\varphi$ be the explicit~$\PSPACE$-machine given by
Lemma~\ref{lem:mcdet} applied to~$\varphi$.  Let~$r_\varphi$
and~$p_\varphi$ be the term and the polynomial also given by that
lemma. By Lemma~\ref{lem:mcdet}.\ref{lem:mcdet.bound2}, the term~$r_\varphi$
witnesses~$M^Y_\varphi$ as explicit~$\EXP$-machine. Therefore,
Lemma~\ref{lem:mcnondet} applies to~$M^Y_\varphi$ and~$r_\varphi$ and we
get an explicit~$\NEXP$-machine~$N_\psi$, a term~$r_\psi$, and a
polynomial~$p_\psi$.  We prove~(a)--(e) using the
quantifier-free~$\PV(\alpha)$-formulas~$F,G,H$ also given by
Lemma~\ref{lem:mcnondet}, and
the~$\Sigma^{1,b}_0$-formula~$\Comp_\varphi$ given by
Lemma~\ref{lem:mcdet}.

For~(a), argue in~$\V^0_2$ and assume~$\psi(\bar x)$ holds. Choose~$Y$
such that~$\varphi(Y,\bar x)$ holds. By Lemma~\ref{lem:mcdet}.\ref{lem:mcdet.c}, the
set~$Z:=\Comp_\varphi(Y,\bar x,\cdot)$ is a halting computation
of~$M^Y_\varphi$ on~$\bar x$. Note that~$Z$ exists
by~$\Sigma^{1,b}_0$-comprehension, which defines the
theory~$\V^0_2$. By Lemma~\ref{lem:mcdet}.\ref{lem:mcdet.b}, the computation~$Z$
cannot be rejecting, so it is accepting. By
Lemma~\ref{lem:mcnondet}.\ref{lem:mcnondet.a}, the set~$F := F(Z,Y,\bar x,\cdot)$ is an
accepting computation of~$N_\psi$ on~$\bar x$. Note that~$F$ exists
by~$\Delta^b_1(\alpha)$-comprehension.

For~(b), argue in~$\S^1_2(\alpha)$ and assume~$Y$ is an accepting
computation of~$N_\psi$ on~$\bar x$. By Lemma~\ref{lem:mcnondet}.\ref{lem:mcnondet.b}
we have that~$G(Y,\bar x,\cdot)$ is an accepting computation
of~$M_\varphi^Z$ on~$\bar x$, for~$Z := H(Y,\bar x,\cdot)$. Note
that~$Z$ exists by~$\Delta^b_1(\alpha)$-comprehension. By
Lemma~\ref{lem:mcdet}.\ref{lem:mcdet.a} we get that~$\varphi(Z,\bar x,\cdot)$ holds.
Thus~$\psi(\bar x)$ follows.

For~(c) and~(d), refer to
Lemma~\ref{lem:mcnondet}.\ref{lem:mcnondet.c}, the choices of~$r_\psi$
and~$p_\psi$, and the fact
that~$\bt_{\psi}(\bar x) = \bt_{\varphi}(\bar x)$. This also gives the
claim that~$r_\psi(\bar x)$ witnesses~$N_\psi$ as
explicit~$\NEXP$-machine.

For~(e), argue in~$\S^1_2(\alpha)$.  If~$\neg\psi(\bar x)$ holds,
use~(b). If~$\psi(\bar x)$ holds, choose~$Y$ such
that~$\varphi(Y,\bar x)$ holds. Then Lemma~\ref{lem:mcdet}.\ref{lem:mcdet.g}
and~$\Delta^b_1(\alpha)$-comprehension imply that there exists an
accepting computation~$Z$ of~$M_\varphi^{Y}$ on~$\bar x$. Now argue as
in~(a).
\end{proof}

\section{Consistency for $\NEXP$} \label{sec:consfornexp}

In this section we define a suitable universal explicit $\NEXP$-machine $M_0$.
We verify the claim from the introduction that both theories
$\{\neg\alpha^c_{M_0}\mid c\ge 1\}$ and $\{\neg\beta^c_{M_0}\mid\penalty10000 c\ge 1\}$
formalize $\NEXP \not\subseteq \Ppoly$.
We finally prove that the
consistency of both formalizations with the theory~$\V^0_2$ follows
from Theorem~\ref{thm:alphaNEXP} and our work on
formally-verified model-checkers.


\subsection{A universal machine} \label{sec:universal}

A canonical~$\NEXP$-complete problem
called $Q_0$ is:
\begin{quote} \emph{Given~$\langle N,x,t\rangle$ as input,
    where~$N$ is a (number coding a) non-deterministic machine, and~$x$
    and~$t$ are numbers written in binary, does~$N$ accept~$x$
    in at most~$t$ steps?}
\end{quote}
A non-deterministic exponential-time machine~$M_0$ for~$Q_0$, on
input~$\langle N,x,t\rangle$, guesses and verifies a time-$t$
computation of~$N$ on~$x$. We ask for an implementation of this so
that a weak theory can verify its correctness.  This is a quite direct
consequence of Lemmas~\ref{lem:mcdet} and~\ref{lem:mc}.

\begin{lemma} \label{lem:M0} There exists an
  explicit $\NEXP$-machine~$M_0$ with one input-tape and without
  oracles, such that for every explicit $\NEXP$-machine~$M$ with one
  input-tape and without oracles, say witnessed by the term~$t_M(x)$,
  there are quantifier-free $\PV(\alpha)$-formulas $F(Z,x,u)$
  and~$G(Z,x,u)$ such that
\begin{enumerate}[label=(\alph*), ref=\alph*] \itemsep=0pt
\item $\begin{array}[t]{lcl}
\S^1_2(\alpha)& \vdash &
    \q{$Z$ is an accepting computation of $M$ on $x$} \to \\
    &&\;\; \q{$F(Z,x,\cdot)$ is an accepting
    computation of $M_0$ on $\langle
    M,x,t_M(x)\rangle$},
\end{array}$
\label{lem:M0.a}
\item $\begin{array}[t]{lcl}
\S^1_2(\alpha)&\vdash&
    \q{$Z$ is an accepting computation of $M_0$ on $\langle M,x,t_M(x)\rangle$} \to \\
    &&\;\; \q{$G(Z,x,\cdot)$ is an accepting computation of $M$ on $x$}.
\end{array}$
\label{lem:M0.b}
\end{enumerate}
In particular,
\begin{enumerate}[resume*]\itemsep=0pt
\item $\begin{array}[t]{lcl}
\S^1_2(\alpha)&\vdash&
  \exists_2 Z \q{$Z$ is an accepting computation of $M_0$ on $\langle M,x,t_M(x)\rangle$}
  \leftrightarrow \\
&&\;\; \exists_2 Z \q{$Z$ is an accepting computation of $M$ on $x$}.
\end{array}$
\label{lem:M0.c}
\end{enumerate}
\end{lemma}

\begin{proof} Let~$\pi_1,\pi_2,\pi_3$ be~$\PV$-functions that
  extract~$x_1,x_2,x_3$ from~$z=\langle
  x_1,x_2,x_3\rangle$. Define~$\Pi^{b}_1$-formulas as follows:
\begin{align*}
& \varphi_1(Z,z)\ :=\ \varphi_2(Z,\pi_1(z),\pi_2(z),\pi_3(z)),\\
& \varphi_2(Z,N,x,t)\ :=\ \q{$Z$ is an accepting time-$t$ computation of $N$ on $x$}.
\end{align*}
Let~$M_1^Z$ be the machine given by Lemma~\ref{lem:mcdet} applied
to~$\varphi_1 = \varphi_1(Z,z)$, and let~$r_1(z)$ be the corresponding
term. Since~$\varphi_1$ is a~$\Pi^b_1(\alpha)$-formula, let~$t_1(z)$
and~$C_1(Z,z,w,u)$ be the term and the
quantifier-free~$\PV(\alpha)$-formula given by
Lemma~\ref{lem:mcdet}.\ref{lem:mcdet.g}.  We set~$M_0$ to the
explicit~$\NEXP$-machine given by Lemma~\ref{lem:mcnondet} applied
to~$M_1^Z$ with term~$r_1(z)$ witnessing it as explicit~$\EXP$-machine by
Lemma~\ref{lem:mcdet}.\ref{lem:mcdet.bound2}. In the proof of (a)--(b)
we use the quantifier-free~$\PV(\alpha)$-formulas~$F_1,G_1,H_1$ given
by Lemma~\ref{lem:mcnondet} on~$M_1^Z$.

For (a) we set $F(Z,x,u):=F_1(C,Z,z,u)$ where $C$
abbreviates~$C_1(Z,z,t_1(z),\cdot)$ and in both cases $z$
abbreviates $\langle M,x,t_M(x)\rangle$. Argue in~$\S^1_2(\alpha)$ and
assume~$Z$ is an accepting computation of~$M$ on~$x$. Since $M$ is
explicit and $t_M(x)$ is a term witnessing it, we have that $Z$ is an
accepting time-$t$ computation of~$M$ on~$x$, for $t := t_M(x)$. It
follows that $\varphi_2(Z,M,x,t_M(x))$ holds, and
hence $\varphi_1(Z,z)$ holds.  Since $\varphi_1$ is
a $\Pi^b_1(\alpha)$-formula, by
Lemma~\ref{lem:mcdet}.\ref{lem:mcdet.g} we have that the
set $C := C_1(Z,z,t_1(z),\cdot)$ is an accepting computation
of~$M_1^Z$ on~$z$. Such a~$C$ exists
by $\Delta^b_1(\alpha)$-comprehension because~$C_1$ is a
quantifier-free $\PV(\alpha)$-formula. By
Lemma~\ref{lem:mcnondet}.\ref{lem:mcnondet.a} we get that the
set $F := F(Z,x,\cdot) = F_1(C,Z,z,\cdot)$ is an accepting computation
of~$M_0$ on~$z$; i.e., the right-hand side of the implication in~(a)
holds. Again, $F$ exists by $\Delta^b_1(\alpha)$-comprehension.

For (b) we set~$G(Z,x,u) := G_1(Z,z,u)$ where, again,~$z$
abbreviates $\langle M,x,t_M(x)\rangle$. Argue in~$\S^1_2(\alpha)$ and
assume $Z$ is an accepting computation of~$M_0$ on~$z$. Then, by
Lemma~\ref{lem:mcnondet}.\ref{lem:mcnondet.b} we have that the
set $G := G(Z,x,\cdot) = G_1(Z,z,\cdot)$ is an accepting computation
of~$M_1^{H}$ on~$z$ for $H := H_1(Z,z,\cdot)$. The two sets $G$
and~$H$ exist by $\Delta^b_1$-comprehension. Now,
Lemma~\ref{lem:mcdet}.\ref{lem:mcdet.a} implies that $\varphi_1(H,z)$
holds; i.e., $H$ is an accepting time-$t$ computation of~$M$ on~$x$,
for $t := t_M(x)$, and hence also an accepting computation of~$M$
on~$x$. This shows that the right-hand side in the implication in~(b)
holds.

The final statement follows from (a) and (b) by
$\Delta^b_1(\alpha)$-comprehension.
\end{proof}

\subsection{Formalization}

The introduction claimed that the theories $\{\neg\alpha^c_{M_0}\mid c\ge 1\}$ and $\{\neg\beta^c_{M_0}\mid c\ge 1\}$ both formalize $\NEXP \not\subseteq \Ppoly$. This is easy to check:

\begin{proposition}\label{prop:betatruth} The following are equivalent.
\begin{enumerate}[label=(\alph*), ref=\alph*,beginpenalty=10000] \itemsep=0pt
\item $\NEXP\not\subseteq\Ppoly$.
\item $\big\{ \neg\alpha^c_{M_0}\mid c \in \N \big\}$ is true.
\item $\big\{ \neg\alpha^c_{M}\mid c\in\N \big\}$ is true for some
  explicit $\NEXP$-machine $M$.
  \item $\big\{ \neg\beta^c_{M_0}\mid c \in \N \big\}$ is true.
  \item $\big\{ \neg\beta^c_{M}\mid c\in\N \big\}$ is true for some
  explicit $\NEXP$-machine $M$.
\end{enumerate}
\end{proposition}

\begin{proof}
  We show that (a)-(b)-(c) are equivalent, and that~(a)-(d)-(e) are
  equivalent. To see that~(a) implies~(b), assume~(b) fails;
  i.e.,~$\alpha^c_{M_0}$ is true for
  some~$c\in\N$. Then~$Q_0\in\SIZE[n^c]$.  As~$Q_0$
  is~$\NEXP$-complete,~(a) fails. That~(b) implies~(c) is trivial
  since~$M_0$ is an explicit~$\NEXP$-machine.  That~(c) implies~(a) is
  obvious since every explicit~$\NEXP$-machine defines a language
  in~$\NEXP$. To see that~(a) implies~(d) argue as in the proof
  that~(a) implies~(b) swapping~$\beta$ for~$\alpha$. That~(d)
  implies~(e) is trivial since~$M_0$ is an
  explicit~$\NEXP$-machine. Finally, that~(e) implies~(a) follows from
  the Easy Witness Lemma~\ref{lem:ewl}.
\end{proof}

It is straightforward to see that the equivalences (b)-(c) and (d)-(e)
in Proposition~\ref{prop:betatruth} have direct proofs (i.e., proofs
that do not rely on the easy witness lemma). We use Lemma~\ref{lem:M0} to prove this
on the formal level, for both formalizations.

\begin{lemma} \label{lem:beta} For every~$c\in\N$ and every~$1$-input
  explicit~$\NEXP$-machine~$M$ without oracles there is~$d\in\N$ such
  that~$\S^1_2(\alpha)$ proves~$(\alpha^c_{M_0} \to \alpha^d_{M})$
  and~$(\beta^c_{M_0}\to \beta^d_M)$.
 \end{lemma}

\begin{proof}
  We refer to the implication
  between~$\alpha$'s as the \emph{$\alpha$-case}, and to the
  implication between~$\beta$'s as the \emph{$\beta$-case}. Both have
  similar proofs, so we prove them at the same time. Let~$M$ be
  witnessed by the term~$t_M(x)$. Let~$F(Z,x,u)$ and~$G(Z,x,u)$ be the
  formulas given by Lemma~\ref{lem:M0} on~$M$. Argue
  in~$\S^1_2(\alpha)$ and assume~$\alpha^c_{M_0}$ or~$\beta^c_{M_0}$,
  as appropriate.  Let~$n\in\Log$ be given. We aim to find a
  circuit~$C$ in the~$\alpha$-case, and two circuits~$C,D$ in
  the~$\beta$-case, witnessing~$\alpha_M^e$ or~$\beta_M^e$,
  respectively, for the given~$n$, and for
  suitable~$e\in\N$. Choose~$d\in\N$ such
  that~$|\langle M,x,t_M(x)\rangle| \ <\ n^{d}$ for all~$x<2^n$. In
  the~$\alpha$-case, let~$C_0$ be a circuit with~$|C_0|<m^c$ that
  witnesses~$\alpha_{M_0}^c$ for~$m := n^d$. In the~$\beta$-case
  let~$C_0,D_0$ be circuits with~$|C_0|,|D_0|<m^c$ that
  witness~$\beta_{M_0}^c$ for~$m:=n^{d}$.

  Choose~$C$ such that~$C(x) = C_0(\langle M,x,t_M(x) \rangle)$
  and~$e\in\N$ such that~$C<2^{n^e}$. This~$C$ will be the
  witness-circuit in the~$\alpha$-case, and the first of the two
  witness-circuits in the~$\beta$-case. For the latter, we choose the
  second circuit~$D$ as follows.  Choose formulas $F,G$ according to
  Lemma~\ref{lem:M0}.
  By  Lemma~\ref{lem:pvcircuit} there is a circuit~$D$ such that
$$
D(x,u) \leftrightarrow G(D_0(\langle M,x,t_M(x)\rangle,\cdot),x,u)
$$
for all~$x,u$ with~$x<2^n$.  Then~$C,D<2^{n^e}$ for
suitable~$e\in\N$. This is the~$e \in \N$ we choose in
the~$\beta$-case.

We claim that~$C$ witnesses~$\alpha^e_M$ for the given~$n$ in
the~$\alpha$-case, and~$C,D$ witness~$\beta^e_{M}$ for the given~$n$
in the~$\beta$-case. Let~$x<2^n$ and
choose~$z := \langle x,M,t_M(x) \rangle$.  Let~$Z$ be any set and
let~$Y := F(Z,x,\cdot)$, which exists
by~$\Delta^b_1(\alpha)$-comprehension.  If~$C(x)=0$, then~$C_0(z)=0$
and both~$\alpha^c_{M_0}$ and~$\beta^c_{M_0}$ imply that~$Y$ is not an
accepting computation of~$M_0$ on~$z$. By
Lemma~\ref{lem:M0}.\ref{lem:M0.a} this means that~$Z$ is not an
accepting computation of~$M$ on~$x$. In both cases, this completes one
half of the verification of the witnesses. If~$C(x)=1$,
then~$C_0(z)=1$ and~$\alpha^c_{M_0}$ implies that there exists an
accepting computation~$Y$ of~$M_0$ on~$z$, and~$\beta^c_{M_0}$ implies
that~$Y := D_0(z,\cdot)$ is such an accepting computation of~$M_0$
on~$z$. But then Lemma~\ref{lem:M0}.\ref{lem:M0.b} implies
that~$Z := G(Y,x,\cdot)$, which exists
by~$\Delta^b_1(\alpha)$-comprehension, is an accepting computation
of~$M$ on~$x$. In both cases, this completes the other half of the
verification of the witness: in the~$\beta$-case,
because~$Z = D(x,\cdot)$.
\end{proof}

\subsection{Consistency}

\noindent For every explicit~$\NEXP$-machine~$M$, which by default has
one input-tape and no oracles, recall
that~$\alpha^c_M := \alpha^c_\psi$ for~$\psi$ as
in Definition~\ref{def:alphaM}. 
For a theory~$\T$ that
extends~$\S^1_2(\alpha)$, consider the following~\emph{A-statements}
for~$\T$:

\begin{tabbing}
\hspace{0.4cm} \= A: \hspace{0.1cm} \= $\T+\{\neg \alpha^c_M\mid c\in\N\}$ is consistent for some
  explicit $\NEXP$-machine $M$, \\
\> A0: \> $\T+\{\neg \alpha^c_{M_0}\mid c\in\N\}$ is consistent.
\end{tabbing}

\noindent Consider also the corresponding \emph{B-statements}
for~$\T$:

\begin{tabbing}
\hspace{0.4cm} \= B: \hspace{0.1cm} \= $\T+\{\neg \beta^c_M\mid c\in\N\}$ is consistent for some
  explicit $\NEXP$-machine $M$, \\
\> B0: \> $\T+\{\neg \beta^c_{M_0}\mid c \in \N\}$ is consistent.
\end{tabbing}

\noindent Next, recall the statement of
Theorem~\ref{thm:alphaNEXP}, which we now state for an arbitrary
theory~$\T$ that extends~$\S^1_2(\alpha)$. We refer to it as
the~\emph{C-statement}, or the \emph{direct consistency statement}
for~$\T$:

\begin{tabbing}
\hspace{0.4cm} \= C: \hspace{0.1cm} \= $\T+\{\neg \alpha^c_{\psi} \mid c \in \N \}$ is consistent
  for some $\hat\Sigma^{1,b}_1$-formula $\psi(x)$.
\end{tabbing}

\noindent Let us explicitly point out that the formula~$\psi(x)$ of
the~C-statement has only one free variable of the number sort, and no
free variables of the set sort.

\begin{lemma}\label{lem:betaalpha}
  For every~$c\in\N$ and every explicit~$\NEXP$-machine~$M$ with one
  input-tape and without oracles,~$\S^1_2(\alpha)$
  proves~$(\beta^c_M\to\alpha^c_M)$.
\end{lemma}

\begin{proof} The formula~$\beta^c_M$ states that the (single)
  existential set-quantifier in~$\alpha^c_M$ is witnessed
  by~$D_x(\cdot)$, and this set exists
  by~$\Delta^b_1(\alpha)$-comprehension.
\end{proof}

We view the following proposition as justification that our
formalization is faithful. It takes record of which
implications in Proposition~\ref{prop:betatruth} hold over weak
theories.

\begin{proposition} \label{prop:statements}
  Let~$\T$ be a theory extending~$\S^1_2(\alpha)$ and consider
  the~A,B,C-statements for~$\T$. Then, the following hold:
  the~A-statements are equivalent, the~B-statements are equivalent,
  and both~A-statements imply both~B-statements as well as
  the~C-statement.
\end{proposition}

\begin{proof} Lemma~\ref{lem:betaalpha} and compactness show that
  each~A-statement implies the corresponding~B-statement.
   Further, Lemma~\ref{lem:beta} proves
  that the~A-statements are equivalent, and that the~B-statements are
  equivalent; for the back implications note that~$M_0$ is certainly
  an explicit~$\NEXP$-machine. Further, it is obvious from the
  definition of~$\alpha^c_M$ that~A implies~C and hence
  both~A-statements imply~C.
\end{proof}

When~$\T = \V^0_2$, we argue below that the model-checker lemmas can
be used to show that the implication A-to-C in
Proposition~\ref{prop:statements} can be reversed. It will follow that
all~A,B,C-statements for~$\V^0_2$ are equivalent. Composing with
Theorem~\ref{thm:alphaNEXP} we get the following corollary, which
entails Theorem~\ref{thm:NEXP}.

\begin{theorem} \label{thm:mainfull}
  For $\T=\V^0_2$ all statements C, A, A0, B, B0 are true.
\end{theorem}

\begin{proof}
  Theorem~\ref{thm:alphaNEXP} states that~C is true
  for~$\T=\V^0_2$. Hence, by Proposition~\ref{prop:statements}, it
  suffices to show that~C implies~A for~$\T=\V^0_2$. But this follows
  from Lemma~\ref{lem:mc}.\ref{lem:mc.a}
  and~\ref{lem:mc}.\ref{lem:mc.b}. Indeed, these state that
  every~$\hat\Sigma^{1,b}_1$-formula~$\psi(x)$ is~$\V^0_2$-provably
  equivalent to~\eqref{eq:comp2} for
  suitable~$M$.
\end{proof}

\section{Consistency for barely superpolynomial time}

In this section we fix $r \in\PV$  such that
\begin{enumerate}\itemsep=0pt
\item[(r0)] the function $x \mapsto r(x)$ is computable in time $O(r(x))$;
\item[(r1)] $\S^1_2\vdash (|x|{=}|y|\to r(x){=}r(y))$;
\item[(r2)] $\S^1_2\vdash  (|x|{<}|y|\to r(x){<}r(y))$;
\item[(r3)] for every polynomial~$p$ there is~$f \in \PV$ such
  that~$\S^1_2 \vdash p(r(x)) \leq r(f(x))$;
\item[(r4)] for every $c\in\N$ there is $n_c\in\N$ such that
$\N\models\forall x\ (|x|{>}n_c\to  r(x){>}|x|^c)$.
\end{enumerate}
\noindent We call a function $r$ satisfying (r4) {\em
  length-superpolynomial}.
An {\em explicit $\NTIME(\poly(r(x)))$-machine} is an
explicit~$\NEXP$-machine~$M$ that is witnessed by~$p(r(x))$ for some
polynomial~$p$.

Here, we deviate from our convention that explicit
machines are witnessed by terms and allow~$\PV$-symbols. In the
notation~$\NTIME(\poly(r(x)))$, the~$x$ is there to emphasize that the
runtime is measured as a function of the input~$x$ and not its
length. If we want to measure runtime as a function of the length of
the input, then we use~$n$ instead of~$x$. For
example,~$\NP = \NTIME(n^{O(1)})$ is given by the collection of
explicit~$\NTIME(\poly(r(x)))$-machines with~$r(x) = |x|$, and the
classes~$\NE = \NTIME(2^{O(n)})$ and~$\NTIME(n^{O(\log^{(k)}n)})$ are
given by the collections of explicit~$\NTIME(\poly(r(x)))$-machines
for~$r(x) = 2^{|x|}$ and~$r(x)=|x|^{\log^{(k)}|x|}$, respectively; the
latter two satisfy (r0)-(r4), if $k \geq 1$ in the second.

\begin{remark} (r3) is not implied by the other conditions.
\end{remark}

\begin{proof}
We shall define a function~$r(x)$ which consists of slow growing
segments interspersed with fast growing segments. First, choose a fast
growing function $R\in\PV$ so that $R(x)$ depends only on~$|x|$ and so
that $R(x)^2 \ge R(x) + |x|^{\omega(1)}$.  For instance $R(x) =
2^{|x|}$ works.  Second, define $\ell:\N\to\N$ be increasing
with $\ell(c+1)>\ell(c)^c+1$ and with $R(x)^2 \ge R(x) + |x|^c$ for
all $x\ge 2^{\ell(c)-1}$.  Let $x_c:=2^{\ell(c)-1}$
and $y_c:=2^{\ell(c)^c}-1$ be the first and last numbers of
length $\ell(c)$ and~$\ell(c)^c$, respectively.  Finally,
let $r(x):=R(x_c)+|x|-|x_c|$ for $x_c\le x\le y_c$, and
let $r(x):=R(x)$ for~$y_c< x< x_{c+1}$.  The slow growing segments
of~$r(x)$ are where $x_c\le x\le y_c$, and here $r(x)$ is chosen to be
as slow growing as possible while satisfying (r1) and~(r2).

Clearly, $\ell$ and~$R$ can be chosen so that $r(x)$ is in $\PV$ and
properties (r0), (r1), (r2), and~(r4) hold for~$r$.  We claim~(r3)
fails for $p(x)=x^2$.

Indeed, let~$f\in\PV$ be given and choose~$c$ such
that~$|f(x_c)|<|x_c|^c=|y_c|$.  Then
\begin{equation*}
p(r(x_c))=r(x_c)^2=R(x_c)^2\ge R(x_c)+|x_c|^{c}=R(x_c)+|y_c|>r(y_c)>r(f(x_c))
\end{equation*}
where the last inequality follows from~(r2).
\end{proof}

\subsection{A more general universal machine}

We start with the analogue of Lemma~\ref{lem:M0}.

\begin{lemma} \label{lem:Mr} There is an
  explicit $\NTIME(\poly(r(x)))$-machine~$M_r$ with one input-tape and
  without oracles such that for every
  explicit $\NTIME(\poly(r(x)))$-machine~$M$ with one input-tape and
  without oracles there are $f_M(x) \in \PV$ and
  quantifier-free $\PV(\alpha)$-formulas $F_M$ and~$G_M$ such that
\begin{enumerate}[label=(\alph*), ref=\alph*,labelindent=0.5em,leftmargin=*] \itemsep=0pt
\item $\begin{array}[t]{lcl}
\S^1_2(\alpha)&\vdash&
      \q{$Z$ is an accepting computation of $M$ on $x$} \to \\
    && \q{$F_M(Z,x,\cdot)$ is an accepting
    computation of $M_r$ on $\langle
    M,x,f_M(x)\rangle$}.
\end{array}$
\label{lem:Mr.a}
\item $\begin{array}[t]{lcl}
 \S^1_2(\alpha)&\vdash&
     \q{$Z$ is an accepting computation of $M_r$ on $\langle M,x,f_M(x)\rangle$} \to \\
    && \q{$G_M(Z,x,\cdot)$ is an accepting computation of $M$ on $x$},
\end{array}$
\label{lem:Mr.b}
\end{enumerate}
In particular,
\begin{enumerate}[resume*,labelindent=0.5em,leftmargin=*]\itemsep=0pt
\item $\begin{array}[t]{lcl}
\S^1_2(\alpha)&\vdash&
  \exists_2 Z \q{$Z$ is an accepting computation of $M_r$ on
  $\langle M,x,f_M(x)\rangle$}
  \leftrightarrow \\
  && \exists_2 Z \q{$Z$ is an accepting computation of $M$ on $x$}
\end{array}$
\label{lem:Mr.c}
\end{enumerate}
\end{lemma}

\begin{proof}
Choose according to Lemma~\ref{lem:mcdet} a machine~$M^Z_\varphi$ and
a term~$r_\varphi(N,x,t)$ for
\begin{align*}
  & \varphi(Z,N,x,t)\ := \q{$Z$ is an accepting time-$t$ computation
    of~$N$ on~$x$}.
\end{align*}
By the comment after Equation~\eqref{eq:formcomp}, there is a
polynomial~$p_1$ so that $\bt_\varphi(N,x,t){\leq}p_1(t,|N|,|x|)$
provably in~$\S^1_2$. By Lemma~\ref{lem:mcdet}.\ref{lem:mcdet.bound1},
there is a polynomial~$p_2$ so
that $r_\varphi(N,x,t){\le}p_2(t,|N|,|x|)$ provably
in~$\S^1_2$. For~$M^Z_\varphi$ choose a machine~$M_1$ and a
term $r_1(N,x,t)$ according to Lemma~\ref{lem:mcnondet}. By
Lemma~\ref{lem:mcnondet}.\ref{lem:mcnondet.c}, there is a
polynomial~$p_3$ so that $r_1(N,x,t)\le p_3(t,|N|,|x|)$.

Define~$M_r$ to compute on~$z$ as follows. It first checks
that~$z=\langle N, x, t\rangle$ for certain~$N,x,t$ and
computes~$\langle N,x,r(t)\rangle$; if the check fails, the machine
stops. After this {\em initial} computation~$M_r$ runs~$M_1$
on~$\langle N,x,r(t)\rangle$.  The initial computation can be
implemented with explicit~$\P$-machines (Lemma~\ref{lem:pvmachine}),
say with time bound~$p_4(|z|)$ for a polynomial~$p_4$.  Then~$M_r$ is
an explicit~$\NTIME(\poly(r(x)))$-machine. Indeed, it is witnessed
by~$p_4(|z|)+p_3(r(z), |z|,|z|)\le p_5(r(z))$ for a
polynomial~$p_5$. Here we use that~$\S^1_2$-provably~$t,N,x$ are
bounded by~$z$, and~$r$ is non-decreasing with~$r(x)\ge|x|$ by~(r1)
and~(r2).

Let~$M$ be an explicit~$\NTIME(\poly(r(x)))$-machine, say witnessed
by~$p_M(r(x))$ for a polynomial~$p_M$. Choose~$f_M$ for~$p_M$
according to~(r3).

For (a), argue in~$\S^1_2$ and assume~$Z$ is an accepting computation
of~$M$ on~$x$. Then~$Z$ is time $p_M(r(x))$, so by~(r3) we can repeat
the halting configuration to get an accepting time $r(f_M(x))$
computation~$Z_0$ of~$M$ on~$x$, i.e., $\varphi(Z_0,M,x,r(f_M(x)))$
holds. By Lemma~\ref{lem:mcdet}.\ref{lem:mcdet.g}, the
set $Z_1:=C_\varphi(Z_0,M,
x,r(f_M(x)),t_\varphi(M,x,r(f_M(x))),\cdot)$ is an accepting
computation of~$M^{Z_0}_\varphi$ on the triple $M, x,r(f_M(x))$. By
Lemma~\ref{lem:mcnondet}.\ref{lem:mcnondet.a}, the
set $Z_2:=F(Z_1,Z_0,M, x,r(f_M(x)),\cdot)$ is an accepting computation
of~$M_1$ on the triple $M,x,r(f_M(x))$.  Compose $Z_2$ with an initial
computation of~$M_r$ on $z:=\langle M,x,f_M(x)\rangle$ to get an
accepting computation~$Z_3$ of~$M_r$ on~$z$. It is clear
that $Z_3=F_M(Z,x,\cdot)$ for some
quantifier-free $\PV(\alpha)$-formula~$F_M$.

For (b), argue in~$\S^1_2$ and let $Z$ be an accepting computation
of~$M_r$ on $\langle M, x, f_M(x)\rangle$. From~$Z$ extract an
accepting computation~$Z_0$ of~$M_1$ on the triple $M,x,r(f_M(x))$. By
Lemma~\ref{lem:mcnondet}.\ref{lem:mcnondet.b}, ~$Z_1:=G(Z_0,M,x,r(f_M(x)),\cdot)$
is an accepting computation of~$M^{Z_2}_\varphi$ on the
triple $M,x,r(f_M(x))$
where $Z_2:=H(Z_0,M,x,r(f_M(x)),\cdot)$. Clearly, $Z_0$ can be
described by a quantifier-free $\PV(\alpha)$-formula, so $Z_1$
and~$Z_2$ exist by $\Delta^b_1(\alpha)$-comprehension. Hence, by
Lemma~\ref{lem:mcdet}.\ref{lem:mcdet.a}, $\varphi(Z_2, M,x,r(f_M(x)))$
holds, i.e., $Z_2$ is an accepting time-$r(f_M(x))$ computation of~$M$
on~$x$.  By~(r3) we can shrink~$Z_2$ to time $p_M(r(x))$ and get an
accepting computation~$Z_3$ of~$M$
on~$x$. Clearly, $Z_3=G_M(Z,x,\cdot)$ for some
quantifier-free $\PV(\alpha)$-formula~$G_M$.

Finally, (c) follows from (a) and (b) by $\Delta^b_1(\alpha)$-comprehension.
\end{proof}

\subsection{Formalization}

To faithfully formalize $\NTIME(\poly(r(x)))\not\subseteq\Ppoly$ we intend to follow the path paved in Section~\ref{sec:consfornexp}. Some modification are, however, required. First, we need an analogue of the Easy Witness Lemma. This has been achieved by Murray and Williams~\cite{MurrayWilliams:CircuitLB}:

\begin{lemma}\label{lem:sewl} Let~$t(n)$ be a function that is increasing,
  time-constructible, and
  superpolynomial. If~$\NTIME(\poly(t(n)))\subseteq\Ppoly$, then
  every~$\NTIME(\poly(t(n)))$-machine~$M$ has po\-ly\-no\-mial-size
  witness circuits.
\end{lemma}

That~$t(n)$ is {\em superpolynomial} means that for every~$c\in\N$
there is $n_c\in\N$ such that $t(n)>n^c$ for all $n>n_c$. That $M$ has
{\em witness circuits of size~$s(n)$}, where $s:\N\to\N$ is a
function, means that for every ~$x\in\{0,1\}^*$ that is accepted
by~$M$, there exists a circuit~$D$ of size at most~$s(|x|)$ such
that~$\mathit{tt}(D)$ encodes an accepting computation of~$M$
on~$x$. Note that, in contrast to Lemma~\ref{lem:ewl}, the circuit~$D$
can depend on~$x$.  We do not know whether Lemma~\ref{lem:sewl} holds
true for {\em oblivious} witness circuits as in Lemma~\ref{lem:ewl}.

Lemma~\ref{lem:sewl} follows from the central result of~\cite{MurrayWilliams:CircuitLB}:

\begin{lemma}[Lemma~4.1 in \cite{MurrayWilliams:CircuitLB}] \label{lem:murraywilliams}
  There are $e,g \in \N$ with $e,g \geq 1$ such that for all
  increasing
  time-constructible functions $s(n)$ and~$t(n)$, and
  for $s_2(n) := s(en)^e$,
  if $\NTIME(O(t(n)^e)) \subseteq \SIZE(s(n))$, then
  every $\NTIME(t(n))$-machine has witness circuits of
  size $s_2(s_2(s_2(n)))^{2g}$, provided that~$s(n) < 2^{n/e}/n$
  and $t(n) \geq\penalty10000 s_2(s_2(s_2(n)))^d$ for a sufficiently
  large $d \in \N$.
\end{lemma}

\begin{proof}[Proof of Lemma~\ref{lem:sewl} from Lemma~\ref{lem:murraywilliams}]
We start noting that there is a non-deterministic machine~$U$ that
decides the problem~$Q_0$ defined in Section~\ref{sec:universal} in
time $O(|x|+|M| \cdot t^2)$ on input $\langle M,x,t\rangle$: after
reading the input, guess the non-deterministic choices of~$M$ and
deterministically in time $c_M \cdot t^2$ simulate the computation
path of~$M$ on input~$x$ as determined by those choices, where~$c_M$
is a simulation overhead constant that depends only on~$M$ and that we
may assume is at most~$|M|$.

Assume $\NTIME(\poly(t(n))) \subseteq \Ppoly$. Fix~$c \in \N$
with $c \geq 1$ and an $\NTIME(t(n)^c)$-machine~$M$. We intend to
apply Lemma~\ref{lem:murraywilliams} to~$M$ for a suitably
chosen~$s(n)$, with~$t(n)^c$ in the role of $t(n)$. For that, we will
need to show that $\NTIME(O(t(n)^{ce})) \subseteq \SIZE(s(n))$ for the
chosen~$s(n)$, where $e \geq 1$ is the first of the two
constants in Lemma~\ref{lem:murraywilliams}.

The restriction of~$U$ to inputs of the
form $\langle M,x,t(|x|)^{ce+1}\rangle$ runs in
time $O(|x|+\penalty10000 |M| \cdot t(|x|)^{2ce+2})$. Therefore, the set of
pairs $\langle M,x\rangle$ such that $U$~accepts on
input $\langle M,x,t(|x|)^{ce+1}\rangle$ is in $\NTIME(\poly(t(n)))$,
so by the assumption, it is decided by circuits of
size $p(|\langle M,x \rangle|)$ for a suitable polynomial~$p(n)$.

  Now, choose~$s(n)$ as a polynomial such that for every
  non-deterministic Turing machine~$M$ and every~$x$ that is
  sufficiently long with respect to~$M$ it holds
  that $p(|\langle M,x \rangle|) < s(|x|)$. We verify
  that $\NTIME(O(t(n)^{ce})) \subseteq \SIZE(s(n))$: if~$B$ is a set
  in $\NTIME(O(t(n)^{ce}))$ and~$M$ is a non-deterministic
  Turing machine that witnesses this, then, for sufficiently long~$x$,
  we have that $x$ is in~$B$ if and only if~$U$ accepts
  on $\langle M,x,t(|x|)^{ce+1}\rangle$. Hence, by the choice
  of~$s(n)$, the set~$B$ is in $\SIZE(s(n))$.

  The requirements of Lemma~\ref{lem:murraywilliams}
  that~$s(n) < 2^{n/e}/n$ and~$t(n)^c \geq s_2(s_2(s_2(n)))^d$ for a
  sufficiently large constant~$d \in \N$ are obviously met
  because~$s(n)$ is polynomially bounded and~$t(n)$ is
  superpolynomial.  Lemma~\ref{lem:murraywilliams} applied to~$s(n)$
  and~$t(n)^c$ then gives that~$M$ has witness circuits of
  size~$s_2(s_2(s_2(n)))^{2g}$, where~$g \geq 1$ is the second of the
  two constants in Lemma~\ref{lem:murraywilliams}. Since~$s(n)$ is
  polynomially bounded, also this function is polynomially
  bounded. Thus,~$M$ has polynomial-size witness circuits.
\end{proof}

Lemma~\ref{lem:sewl} enables a $\forall\Pi^{1,b}_1$-formalization
of $\NTIME(\poly(r(x)))\not\subseteq\Ppoly$:

\begin{definition}
For an explicit~$\NTIME(\poly(r(x)))$-machine~$M$ with one input-tape
and without oracles define
\begin{equation*}
\begin{array}{lcl}
\gamma_M^c&:=&\forall n{\in}\Log\  \exists C{<}2^{n^c}\
\forall x{<}2^n\ \exists D{<}2^{n^c}\  \forall_2 Y \\
&&\quad (C(x){=}0\ \to\ \neg\q{$Y$ is an accepting
computation of $M$ on $x$})\ \wedge\\
&&\quad (C(x){=}1\ \to\ \q{$D(\cdot)$ is an accepting computation
of $M$ on $x$}).
\end{array}
\end{equation*}
Let~$M_r$ be the explicit~$\NTIME(\poly(r(x))))$-machine of
Lemma~\ref{lem:Mr}. Define
\begin{equation*}
\q{$\NTIME(\poly(r(x)))\not\subseteq\Ppoly$} \ := \ \big\{
\neg\gamma^c_{M_r}\mid c\in\N \big\}.
\end{equation*}
\end{definition}

The following is the analogue of Lemma~\ref{lem:betaalpha} and is
similarly proved.

\begin{lemma}\label{lem:gammaalpha}
  For every~$c\in\N$ and every
  explicit~$\NTIME(\poly(r(x)))$-machine~$M$ with one input-tape and
  without oracles,~$\S^1_2(\alpha)$
  proves~$(\gamma^c_M\to\alpha^c_M)$.
\end{lemma}

\begin{lemma} \label{lem:gamma}
For every $c\in\N$ and every
explicit $\NTIME(\poly(r(x)))$-machine~$M$ with one input-tape and
without oracles there is $d\in\N$ such that $\S^1_2(\alpha)$
proves $(\alpha^c_{M_r} \to \alpha^d_{M})$ and $(\gamma^c_{M_r}\to
\gamma^d_M)$.
 \end{lemma}

\begin{proof} This is proved similarly as Lemma~\ref{lem:beta}. We only treat the $\gamma$-case. Choose $f_M(x)\in\PV$ according to Lemma~\ref{lem:Mr}.
Argue in~$\S^1_2(\alpha)+\gamma^c_{M_r}$. Let~$n \in \Log$ be given.
Choose~$e\in\N$ such that~$|\langle M,x,f_M(x)\rangle| < n^{e}$ for
all~$x < 2^n$. Choose~$C_0$ witnessing~$\gamma^c_{M_r}$
for $m:=n^{e}$.  Choose a circuit~$C$ such that $C(x)=C_0(\langle
M,x,f_M(x)\rangle)$ for all $x < 2^n$.  We shall choose~$d$ large enough
such that $C \le 2^{n^d}$ and choose~$C$ to witness the first
existential quantifier in~$\gamma^d_{M}$ for~$n$.  To verify this
choice, let $x < 2^n$ be given.

  If~$C(x) = 0$, then there are no accepting computations of~$M_r$
  on~$\langle M,x,f_M(x)\rangle$. By
  Lem\-ma~\ref{lem:Mr}.\ref{lem:Mr.a}
  and~$\Delta^b_1(\alpha)$-comprehension, there are no accepting
  computations of~$M$~on~$x$. If~$C(x) = 1$, then there is a
  circuit~$D_0 < 2^{m^c}$ such that~$D_0(\cdot)$ is an accepting
  computation of~$M_r$ on~$\langle M,x,f_M(x)\rangle$. By
  Lemma~\ref{lem:Mr}.\ref{lem:Mr.b},~$G_M(D_0(\cdot),x,\cdot)$ is an
  accepting computation of~$M$ on~$x$.  By Lemma~\ref{lem:pvcircuit}
  there is a circuit~$D$ such that~$\big(D(u)\leftrightarrow
  G_M(D_0(\cdot),x,u)\big)$ for all~$u \le \langle
  p_M(r(x)),p_M(r(x),|M|)\rangle$ where~$p_M$ is a polynomial such
  that~$p_M(r(x))$ witnesses~$M$. Choose~$d\in\N$ large enough such
  that~$D < 2^{n^d}$.
\end{proof}

Finally, we are in the position to verify that the formulas considered
formalize the intended circuit lower bound.

\begin{proposition}\label{prop:gammatruth}
The following are equivalent.
\begin{enumerate}[label=(\alph*), ref=\alph*] \itemsep=0pt
\item $\NTIME(\poly(r(x)))\not\subseteq\Ppoly$.
\item $\big\{ \neg\alpha^c_{M_r}\mid c \in \N \big\}$ is true.
\item $\big\{ \neg\alpha^c_{M}\mid c\in\N \big\}$ is true for some
  explicit $\NTIME(\poly(r(x)))$-machine $M$.
  \item $\big\{ \neg\gamma^c_{M}\mid c\in\N \big\}$ is true for some
  explicit $\NTIME(\poly(r(x)))$-machine $M$.
  \item $\big\{ \neg\gamma^c_{M_r}\mid c \in \N \big\}$ is true.
\end{enumerate}
\end{proposition}

\begin{proof}
  To see that (a) implies~(b), assume~(b) fails, so~$\alpha^c_{M_r}$
  is true for some~$c\in\N$.  Then the problem accepted by~$M_r$ is
  in $\SIZE[n^c]$.  By Lemma~\ref{lem:Mr} this problem
  is $\NTIME(\poly(r(x)))$-hard under polynomial time reductions.
  Since $\Ppoly$ is
  downward-closed under polynomial-time reductions, (a) fails.  The
  claim that (b) implies~(c) is trivial since~$M_r$ is an
  explicit $\NTIME(\poly(r(x)))$-machine.  That (c) implies~(d)
  follows from Lemma~\ref{lem:gammaalpha}.  That~(d) implies~(e)
  follows from Lemma~\ref{lem:gamma}. That (e) implies~(a) follows
  from Lemma~\ref{lem:sewl}: by~(r1) there is a function~$t(n)$ such
  that $t(|x|)=r(x)$ for every~$x$;
  then $\NTIME(\poly(r(x)))=\NTIME(\poly(t(n)))$ where the time-bound
  on the left is written as a function of the input~$x$ and on the
  right as a function of its length $n=|x|$; further, $t(n)$ is
  time-constructible by (r0) and~(r1),
  increasing by~(r2) and superpolynomial by~(r4).
\end{proof}

\subsection{Consistency}

\noindent For a theory~$\T$ that extends~$\S^1_2(\alpha)$, the
new A,B-statements are the following:
\begin{tabbing}
\hspace{0.3cm} \= A$_r$: \hspace{0.1cm} \= $\T+\{\neg \alpha^c_M$ \= $\mid c\in\N\}$ is consistent for some explicit $\NTIME(\poly(r(x)))$-machine $M$, \\
\> B$_r$: \hspace{0.1cm} \> $\T+\{\neg \gamma^c_M$ \> $\mid c\in\N\}$ is consistent for some explicit $\NTIME(\poly(r(x)))$-machine $M$, \\
\> A0$_r$: \> $\T+\{\neg \alpha^c_{M_r}$ \> $\mid c\in\N\}$ is consistent. \\
\> B0$_r$: \> $\T+\{\neg \gamma^c_{M_r}$ \> $\mid c \in \N\}$ is consistent.
\end{tabbing}
To define the corresponding C-statement, we say that the bounding term
of a~$\hat\Sigma^{1,b}_1$-formula~$\psi = \psi(x)$ is \emph{polynomial
  in}~$r(x)$ if~$\S^1_2$
proves~$\bt(\psi) \leq p(r(x))$ for some polynomial~$p(n)$. Then:
\begin{tabbing}
  \hspace{0.3cm} \= C$_r$: \hspace{0.1cm}
  \= $\T+\{\neg \alpha^c_{\psi}$ \=~$\mid c \in \N\}$ is consistent for
  some~$\hat\Sigma^{1,b}_1$-formula~$\psi=\psi(x)$ whose \\
  \> \> bounding term is polynomial in~$r(x)$.
\end{tabbing}

Before we prove the analogue of Theorem~\ref{thm:mainfull} we state
the proof complexity lower bound on which it is based. Recall the
Pigeonhole Principle formula~$\PHP(x)$ from the proof of
Theorem~\ref{thm:alphaNEXP}. The first strong lower bounds on the
provability of~$\PHP(x)$ were due to Ajtai~\cite{Ajtai:PHP}; here we need
the later quantitative improvements from~\cite{BIKPPW:PHP}. This can be
called the~\emph{gem} of proof complexity. We use it in the following
form. Recall that a function is called length-superpolynomial when it
satisfies~(r4).

\begin{theorem}[Gem Theorem] \label{thm:ABIPKPW} For
  every length-superpolynomial $\PV$-function $s(x)$,
  the theory~$\V^0_2$ does not
  prove~$\PHP(s(x))$.
\end{theorem}

\begin{proof}
Consider the Paris-Wilkie propositional translations $F_n :=
\langle\PHP(s(n))\rangle_n$ for $n \in \N$;
see \cite[Definition~9.1.1]{Krajicek:book} in the form used
in \cite[Corollary~9.1.4]{Krajicek:book}. Assume for contradiction
that $\PHP(s(x))$ is provable in~$\V^0_2$. Then, there exist
constants $c,d \in \N$ such that for every sufficiently large
$n \in \N$, the propositional formulas~$F_n$ have Frege proofs of depth~$d$
and size~$2^{|n|^c}$: apply \cite[Corollary~9.1.4]{Krajicek:book} with the
function $f(x) = x\#x$ and note that~$\V^0_2$ is conservative over the
theory considered there: from a model of that theory, get a model
of~$\V^0_2$ by just adding all bounded sets that are definable by
bounded formulas.

  Now, let~$n \in \N$ be large enough to ensure this upper bound and
  at the same time such that~$s(n) > |n|^{6^d c}$, which exists
  because~$s(x)$ is length-superpolynomial.  Setting~$m := s(n)$, this
  means that the propositional formula~$\PHP^{m+1}_m := F_n$ has Frege
  proofs of depth~$d$ and size bounded by an exponential
  in~$m^{1/6^d}$. It is well-known that if~$m$ is sufficiently large,
  then this is false; see \cite[Theorem~12.5.3]{Krajicek:book}.
\end{proof}

Finally we can prove the analogue of Theorem~\ref{thm:mainfull}, which
entails Theorem~\ref{thm:NTIME}.

\begin{theorem}
  For~$\T = \V^0_2$, all statements C$_r$,
  A$_r$, A0$_r$, B$_r$, B0$_r$ are true.
\end{theorem}

\begin{proof}
  The analogue of Proposition~\ref{prop:statements} for
  the~A$_r$,B$_r$,C$_r$-statements has the same proof using
  Lemmas~\ref{lem:gammaalpha}, \ref{lem:gamma} in place of
  Lemmas~\ref{lem:betaalpha}, \ref{lem:beta}. Note that the claim
  that~A$_r$ implies~C$_r$ follows from the remark after
  Equation~\eqref{eqn:referredtolater}.  As in the proof of
  Theorem~\ref{thm:mainfull}, that~C$_r$ implies A$_r$
  for~$\T = \V^0_2$ follows from Lemma~\ref{lem:mc}.\ref{lem:mc.a}
  and~\ref{lem:mc}.\ref{lem:mc.b}. We also
  need~\ref{lem:mc}.\ref{lem:mc.c} along with~$r(x)\ge|x|$ by~(r1)
  and~(r2) to guarantee that the explicit~$\NEXP$-machine is an
  explicit~$\NTIME(\poly(r(x)))$-machine.

We are left to show that C$_r$ holds
  for~$\T=\V^0_2$. This is proved by tightening the choice of
  parameters in the argument that proved
  Theorem~\ref{thm:alphaNEXP}.

Consider the formula
  \begin{equation}
  y{\leq}r(x) \wedge \neg\PHP(y)
\end{equation}
and write this as $\psi = \psi(z)$, where $z = \langle x,y \rangle$;
i.e., $x = \pi_1(z)$ and $y = \pi_2(z)$ with $\pi_1$ and~$\pi_2$
as $\PV$-functions. The formula~$\psi(z)$ is logically equivalent to
a $\hat\Sigma^{1,b}_1$-formula whose bounding term is polynomial
in~$r(z)$ by (r1) and~(r2). We
claim that $\V^0_2 + \{\neg\alpha^c_{\psi} \mid c \in \N\}$ is
consistent, which will give~C$_r$.

  For the sake of contradiction, assume otherwise. By compactness,
  there exists $c \in \N$ such that~$\V^0_2$
  proves~$\alpha^c_{\psi}$. As in the proof of
  Theorem~\ref{thm:alphaNEXP}, we show that this implies
  that $\V^0_2$ proves~$\PHP(r(x))$, which contradicts the Gem
  Theorem by~(r4).

  Argue in~$\V^0_2$ and set~$n := \max\{|z|,2\}$,
  where~$z = \langle x,r(x) \rangle$. Then~$\alpha^c_{\psi}$ on~$n$
  gives a circuit~$C$ such that, for all~$u{\leq} z$ and $v{\leq}z$
  with~$\langle u,v\rangle{\leq}z$, we have
  \begin{equation*}
    \neg C(\langle u,v \rangle) \leftrightarrow (v{\leq}r(u) \to \PHP(v)).
  \end{equation*}

  \noindent Noting that~$\langle x,v\rangle{\leq}z$ for all~$v{\leq}r(x)$,
  fix~$u$ to~$x$ in the circuit~$C(\langle u,v \rangle)$ and get a
  circuit~$D(v)$ such that
  \begin{equation*}
  \forall v{\leq}r(x)\ (\neg D(v) \leftrightarrow \PHP(v)).
  \end{equation*}
Recall that~$\V^0_2$ proves that~$\PHP(x)$ is inductive.  Hence,
plugging~$\neg D(v)$ for~$\PHP(v)$ gives~$\PHP(r(x))$ by
quantifier-free~$\PV(\alpha)$-induction.
\end{proof}

\section{Magnification}\label{sec:magnification}

For this section, a~$\exists_2\Pi^b_1(\alpha)$-formula is
a $\hat\Sigma^{1,b}_1$-formula as in~\eqref{eqn:generichatSigma} in
which its
maximal $\Sigma^{1,b}_0$-subformula~$\varphi(\bar X,Y,\bar x)$ is
a $\Pi^b_1(\alpha)$-formula.

\begin{lemma}\label{lem:cprh}
  For every~$c\in\N$ and
  every~$\exists_2\Pi^b_1(\alpha)$-formula~$\psi(\bar x,y)$
  without free
  set variables,
  the theory~$\S^1_2(\alpha)+\beta^c_{M_0}$ proves
\begin{equation}\label{eq:CA}
\exists C\ \forall y{\le}z\ \big( C(y){=}1\leftrightarrow  \psi(\bar x,y)\big).
\end{equation}
\end{lemma}

\begin{proof} Argue in~$\S^1_2(\alpha)+\beta^c_{M_0}$.  For simplicity
  assume~$\bar x$ is empty.
   For~$\psi = \psi(y)$ choose~$M:=N_\psi$
  according to Lemma~\ref{lem:mc}. Note that since~$\psi$ does not
  have free set variables,~$M$ is without oracles.  By
  Lemma~\ref{lem:mc}.\ref{lem:mc.e}, the formula~$\psi(y)$ is
  equivalent to
\begin{equation*}
\exists_2 Y \q{$Y$ is an accepting computation of $M$ on $y$}.
\end{equation*}
By Lemmas \ref{lem:betaalpha} and~\ref{lem:beta} we
have~$\alpha^d_{M}$ for some $d\in\N$. Let~$z$ be given and
choose $n\in\Log$ with $|z|\le n$. Let $C$ witness $\alpha^d_{M}$
for~$n$. This $C$ witnesses \eqref{eq:CA}.
\end{proof}

It follows that over~$\S^1_2(\alpha)$ the circuit upper bound
statement~$\beta^c_{M_0}$ implies comprehension
for~$\exists_2\Pi^b_1(\alpha)$-formulas {\em without free set
  variables}. For later reference, we note that allowing free set
variables entails full~$\hat\Sigma^{1,b}_1$-comprehension:

\begin{lemma} \label{lem:specialcomp}
$\S^1_2(\alpha) + \exists_2\Pi^b_1(\alpha)$-comprehension proves
  $\V^1_2$.
\end{lemma}

\begin{proof}
  Let~$\T$ denote~$\S^1_2(\alpha) +
  \exists_2\Pi^b_1(\alpha)$-comprehension.  Since $\S^1_2(\alpha) +
  \Sigma^{1,b}_1$-comprehension proves~$\V^1_2$, it suffices to show
  that the set of formulas that are $\T$-provably equivalent to
  an $\exists_2\Pi^b_1(\alpha)$-formula is closed under $\vee$, $\wedge$,
  $\exists_2 Y$, $\exists y{\le}t(\bar x)$ and $\forall y{\le}t(\bar
  x)$. We verify the latter: the formula
$$
\forall y{\le}u\ \exists_2 Y\ \varphi(\bar X,Y,\bar x,u,y)
$$
with~$\varphi(\bar X,Y,\bar x,u,y)$ a~$\Pi^b_1(\alpha)$-formula
is~$\T$-provably equivalent to
$$
\exists_2 Z\ \forall y{\le}u\ \varphi(\bar X, Z(y,\cdot),\bar x,u,y),
$$
where~$Z(y,v)$ abbreviates the atomic
formula $\langle y,v\rangle\in Z$.  Indeed, assuming the former
formula, the latter is proved by induction on~$u$. As the latter is
an $\exists_2\Pi^b_1(\alpha)$-formula, induction for it follows from
comprehension.
\end{proof}

The following lemma makes precise the idea sketched in Section~\ref{sec:introcomprehension}.

\begin{lemma} \label{lem:V12} For every $c\in\N$ and every
  model $(M,\mathcal X)$ of $S^1_2(\alpha)+\beta^c_{M_0}$, there
  exists $\mathcal Y\subseteq\mathcal X$ such that $(M,\mathcal Y)$ is
  a model of~$\V^1_2$.
\end{lemma}

\begin{proof}
  By~$\Delta_1^b(\alpha)$-comprehension, for every~$C\in M$ that is a
  circuit in the sense of~$M$ there is a set~$A\in \mathcal X$ such
  that
$$
(M,\mathcal X)\models\forall y\ (C(y){=}1\leftrightarrow y{\in}A).
$$
By extensionality such a set~$A$ is uniquely determined by~$C$ and we
write~$\hat C$ for it. For these two claims we used the fact
that~$C(y){=}1 \to y{<}2^{|C|}$ holds in every model of $\S^1_2$.

Let
\begin{equation*}
 \mathcal Y:=\big\{ \hat C\in \mathcal X\mid C\in M  \textit{ is a circuit in the sense of }M\big\}.
\end{equation*}
Since~$\mathcal Y \subseteq \mathcal X$, the model~$(M,\mathcal Y)$
satisfies all~$\Pi^{1,b}_1$-sentences which are true
in~$(M,\mathcal X)$, so in particular extensionality, set
boundedness,~$\Sigma^b_1(\alpha)$-induction, and~$\beta^c_{M_0}$.

The point of the model~$(M,\mathcal Y)$ is that it eliminates set
parameters. More precisely, let~$\varphi(\bar x)$ be
a~$\Sigma^{1,b}_\infty$-formula with parameters from~$(M,\mathcal Y)$,
and define~$\varphi^*(\bar x)$ as follows: replace every subformula of
the form~$t{\in}\hat C$ where~$t$ is a term (possibly with number
parameters from~$M$) and~$\hat C$ is a set parameter from~$\mathcal Y$
by~$C(t){=}1$ (i.e., by~$\eval(C,t){=}1$). Note every set parameter
in~$\varphi(\bar x)$ becomes a number parameter
in~$\varphi^*(\bar x)$, and
\begin{equation}\label{eq:star}
(M,\mathcal Y)\models\forall\bar x\ (\varphi(\bar x)\leftrightarrow\varphi^*(\bar x)).
\end{equation}

\noindent {\em Claim:} $(M,\mathcal Y)\models \S^1_2(\alpha)$.\medskip

\noindent{\em Proof of the Claim.} It suffices to show
that $(M,\mathcal Y)$ models $\Delta^{b}_1(\alpha)$-comprehension. So
let $\varphi(x)$ be a $\Delta^{b}_1(\alpha)$-formula with parameters
from $(M,\mathcal Y)$ and $a\in M$. Then $\varphi^*(x)$ is a
number-sort formula, namely a $\Delta^b_1$-formula with (number)
parameters from~$M$. Since $M\models\S^1_2$, Buss' witnessing theorem
implies that~$\varphi^*(x)$ is equivalent in~$M$ to a
quantifier-free $\PV$-formula with the same parameters.
Lemma~\ref{lem:pvcircuit} applied to $n:= \max\{|a|,2\}$ gives a
circuit~$C$ in the sense of~$M$ such that
\[
M\models \forall x{<}2^n (C(x)=1\leftrightarrow\varphi^*(x)).
\]
Then
$\hat C\in\mathcal Y$ and $(M,\mathcal Y)$ satisfies
$\forall y{\le}a ( y\in \hat C\leftrightarrow\varphi(y))$ by
\eqref{eq:star}.\hfill$\dashv$\medskip

By the Claim and Lemma~\ref{lem:specialcomp}, it suffices to show
that~$(M,\mathcal Y)$
has~$\exists_2\Pi^b_1(\alpha)$-comprehension. Let~$\psi(x)$
be a~$\exists_2\Pi^b_1(\alpha)$-formula with parameters
from~$(M,\mathcal Y)$, and let~$a\in M$.  Then~$\psi^*(x)$ is
a~$\exists_2\Pi^b_1(\alpha)$-formula without set parameters.  We
already noted that~$(M,\mathcal Y)\models\beta^c_{M_0}$. Hence, by the
Claim, Lemma~\ref{lem:cprh} applies and gives~$C\in M$ such that
$$
(M,\mathcal Y)\models \forall x{\le}a\,(C(x){=}1\leftrightarrow\psi^*(x)).
$$
Then $\hat C\in\mathcal Y$ and $(M,\mathcal Y)$ satisfies
$\forall x{\le}a\,( x{\in}\hat C\leftrightarrow\psi(x))$ by \eqref{eq:star}.
\end{proof}

As announced in Section~\ref{sec:introcomprehension} this lemma
implies Theorems~\ref{thm:conservative} and~\ref{thm:mag}.

\begin{proof}[Proof of Theorem~\ref{thm:conservative}] Assume
  that $\T$  is inconsistent
  with $\q{$\NEXP\not\subseteq\Ppoly$}$. By compactness, $\T$ proves~$\beta^c_{M_0}$ for some $c\in\N$. Let $\psi$
  be a number sort consequence of~$\V^1_2$ and $(M,\mathcal X)$ a
  model of~$\T$. We have to show that $M\models\psi$. But by
  Lemma~\ref{lem:V12} there exists $\mathcal Y \subseteq \mathcal X$
  such that $(M,\mathcal Y)\models\V^1_2$,
  so $(M,\mathcal Y)\models\psi$, and $M\models\psi$.
\end{proof}

\begin{proof}[Proof of Theorem~\ref{thm:mag}]
  Assume $\S^1_2(\alpha)$ does not prove
  \q{$\NEXP\not\subseteq\Ppoly$}, say, it does not
  prove $\neg \beta^c_{M_0}$. Then there is a model $(M,\mathcal X)$
  of $\S^1_2(\alpha)+\beta^c_{M_0}$. By Lemma~\ref{lem:V12} there
  exists $\mathcal Y\subseteq\mathcal X$ such
  that $(M,\mathcal Y)\models\V^1_2$.  Since $\beta^c_{M_0}$ is
  a $\Pi^{1,b}_1$-formula, we
  have $(M,\mathcal Y)\models\beta^c_{M_0}$. Thus, $\V^1_2$~does not
  prove \q{$\NEXP\not\subseteq\Ppoly$}.
\end{proof}

\begin{remark} \label{rem:hope} The introduction mentioned that
  Theorem~\ref{thm:mag} might raise hopes to complete Razborov's
  program by construcing a model of~$\S^1_2(\alpha)$ satisfying
  some~$\beta^c_{M_0}$. There are good general methods to construct
  models even of certain extensions of~$\T^1_2(\alpha)$ based on
  forcing (see \cite{Riis:Finitization} and~\cite{Muller:TypicalForcing} for an
  extension). However, these methods are tailored
  for $\hat \Sigma_1^{1,b}(\alpha)$-statements, not~$\Pi_1^{1,b}$
  like~$\beta^c_{M_0}$. By the method of feasible interpolation and
  assuming the existence of suitable pseudorandom generators,
  Razborov~\cite{Razborov:unprovability} proved that for
  every $\Sigma^{b}_\infty$-definable $t(n) = n^{\omega(1)}$ and
  every $\Sigma^{b}_{\infty}$-formula $\varphi(x)$ there exists a
  model~$(M,\mathcal X)$ of~$\S^2_2(\alpha)$ that for some $n\in M$
  contains a set $C\in\mathcal X$ coding a size-$t(n)$ circuit that
  computes~$\varphi(x)$; i.e., for every $a<2^n$ there
  is $X_a\in\mathcal X$ coding a computation of~$C$ on~$a$ of the
  truth value of~$\varphi(a)$. Getting a circuit (and computations)
  coded by a number seems to require new ideas.

The best currently known unprovability
result is due to Pich~\cite[Corollary 6.2]{Pich:CircuitLB} and is conditional:
a theory formalizing $\mathsf{NC}^1$-reasoning does not prove  almost everywhere superpolynomial lower bounds for $\textsf{SAT}$ unless subexponential-size formulas can approximate polynomial-size circuits. Reaching $\S^1_2$ seems to require new ideas.
\end{remark}

\ifdefined\JMLversion
\section*{Acknowledgments}
A.\ Atserias was supported in part by Project PID2019-109137GB-C22 (PROOFS)
and the Severo Ochoa and María de Maeztu Program for Centers and Units of Excellence
in R\&D (CEX2020-001084-M) of the Spanish State Research Agency.
S.\ Buss was supported in part by Simons Foundation grant 578919.
\fi 

\ifdefined\JMLversion
\bibliographystyle{ws-jml}
\else 
\bibliographystyle{siam}
\fi
\bibliography{logic}

\end{document}